\documentclass{llncs}

\usepackage{rotating}
\usepackage{graphicx}
\usepackage{bussproofs}
\usepackage{amsmath}
\usepackage{amssymb}

\def\SB#1{\textsubscript{{#1}}}

\def \BF #1{\textbf{{#1}}}
\newenvironment{bprooftree}
{\leavevmode\hbox\bgroup}
{\DisplayProof\egroup}

\title{Explicit Computational Paths}

\author{Arthur F. Ramos \and Ruy J. G. B. de Queiroz \and Anjolina G. de Oliveira}

\institute{Centro de Inform\'atica\\
	Universidade Federal de Pernambuco\\
	\email{afr@cin.ufpe.br}\\
	\email{ruy@cin.ufpe.br}\\
	\email{ago@cin.ufpe.br}\\
}

\begin{document}

	\maketitle
	
	\begin{abstract}
		The treatment of equality as a type in type theory gives rise to an interesting type-theoretic structure known as `identity type'. The idea is that, given terms $a,b$ of a type $A$, one may form the type $Id_{A}(a,b)$, whose elements are proofs that $a$ and $b$ are equal elements of type $A$. A term of this type, $p : Id_{A}(a,b)$, makes up for the grounds (or proof) that establishes that $a$ is indeed equal to $b$. Based on that, a proof of equality can be seen as a sequence of substitutions and rewrites, also known as a `computational path'. One interesting fact is that it is possible to rewrite computational paths using a set of reduction rules arising from an analysis of redundancies in paths. These rules were mapped by De Oliveira in 1994 in a term rewrite system known as $LND_{EQ}-TRS$. Here we use computational paths and this term rewrite system to develop the main foundations of homotopy type theory, i.e., we develop the lemmas and theorems connected to the main types of this theory, types such as products, coproducts, identity type, transport and many others. We also show that it is possible to directly construct path spaces through computational paths. To show this, we construct the natural numbers and the fundamental group of the circle, showing results connected to these structures.
	\end{abstract}

	\section{Introduction}\label{intro}

There seems to be little doubt that the identity type is one of the most intriguing concepts of  Martin-L\"of's Type Theory. This claim is supported by recent groundbreaking discoveries. In 2005, Vladimir Voevodsky \cite{Vlad1} discovered the Univalent Models, resulting in a new area of research known as homotopy type theory \cite{Steve1}. This theory is based on the fact that a term of some identity type, for example $p: Id_{A}(a,b)$, has a clear homotopical interpretation. The interpretation is that the witness $p$ can be seen as a homotopical path between the points $a$ and $b$ within a topological space $A$. This simple interpretation has made clear the connection between type theory and homotopy theory, generating groundbreaking results, as one can see in \cite{hott,Steve1}. Nevertheless, it is important to emphasize that 
the homotopic paths exist only in the semantic sense. In other words, there is no formal entity in type theory that represents these paths. They are not present in the syntax of type theory.

In this work, we are interested in an entity known as computational path, originally proposed by \cite{Ruy4}. A computational path is an entity that establishes the equality between two terms of the same type. It differs from the homotopical path, since it is not only a semantic interpretation. It is a formal entity of the equality theory. In fact, we proposed in \cite{Ruy1} that it should be considered as the type of the identity type. Moreover, we have further developed this idea in \cite{Art3}, where we proposed a groupoid model and proved that computational paths also refute the uniqueness of identity proofs. Thus, we obtained a result  that is on par with the same one obtained by Hofmann \& Streicher (1995) for the original identity type \cite{hofmann2}.

Our main idea in this work is to develop further our previous results. Specifically, we want to focus on the foundations of homotopy type theory. Our objective is to develop the main building blocks of this theory using computational paths. To do this, we prove quite a few lemmas and theorems of homotopy type theory involving the basic types, such as products, coproducts, transport, etc. We thus proceed to show that computational paths can be directly used to simulate path spaces. We argue that it is one of the main advantages of our approach, since it avoids the use of complicated techniques such as the code-encode-decode one. To illustrate that, we work with the natural numbers and with the fundamental group of the circle, showing how one can construct these structures through computational paths.

This work is structured as thus: in the sections 2, 3 and 4, we review the concept of computational paths and its connection to the identity type in type theory. In section 5, we use computational paths to establish the foundations of homotopy type theory. Since sections 2, 3 and 4 are only brief introductions to the theory of computational paths, we refer to papers \cite{Ruy1} and \cite{Art3} for a thoroughly introduction to this subject. 

\section{Computational Paths} \label{path}

Since computational path is a generic term, it is important to emphasize the fact that we are using the term computational path in the sense defined by \cite{Ruy5}. A computational path is based on the idea that it is possible to formally define when two computational objects $a,b : A$ are equal. These two objects are equal if one can reach $b$ from $a$ applying a sequence of axioms or rules. This sequence of operations forms a path. Since it is between two computational objects, it is said that this path is a computational one. Also, an application of an axiom or a rule transforms (or rewrite) an term in another. For that reason, a computational path is also known as a sequence of rewrites. Nevertheless, before we define formally a computational path, we can take a look at one famous equality theory, the $\lambda\beta\eta-equality$ \cite{lambda}:

\begin{definition}
	The \emph{$\lambda\beta\eta$-equality} is composed by the following axioms:
	
	\begin{enumerate}
		\item[$(\alpha)$] $\lambda x.M = \lambda y.M[y/x]$ \quad if $y \notin FV(M)$;
		\item[$(\beta)$] $(\lambda x.M)N = M[N/x]$;
		\item[$(\rho)$] $M = M$;
		\item[$(\eta)$] $(\lambda x.Mx) = M$ \quad $(x \notin FV(M))$.
	\end{enumerate}
	
	And the following rules of inference:

	\bigskip
	\noindent
	\begin{bprooftree}
		\AxiomC{$M = M'$ }
		\LeftLabel{$(\mu)$ \quad}
		\UnaryInfC{$NM = NM'$}
	\end{bprooftree}
	\begin{bprooftree}
		\AxiomC{$M = N$}
		\AxiomC{$N = P$}
		\LeftLabel{$(\tau)$}
		\BinaryInfC{$M = P$}
	\end{bprooftree}
	
	\bigskip
	\noindent
	\begin{bprooftree}
		\AxiomC{$M = M'$ }
		\LeftLabel{$(\nu)$ \quad}
		\UnaryInfC{$MN = M'N$}
	\end{bprooftree}
	\begin{bprooftree}
		\AxiomC{$M = N$}
		\LeftLabel{$(\sigma)$}
		\UnaryInfC{$N = M$}
	\end{bprooftree}
	
	\bigskip
	\noindent
	\begin{bprooftree}
		\AxiomC{$M = M'$ }
		\LeftLabel{$(\xi)$ \quad}
		\UnaryInfC{$\lambda x.M= \lambda x.M'$}
	\end{bprooftree}
	
	
	
	
\end{definition}





\begin{definition}( \cite{lambda})
	$P$ is $\beta$-equal or $\beta$-convertible to $Q$  (notation $P=_\beta Q$)
	iff $Q$ is obtained from $P$ by a finite (perhaps empty)  series of $\beta$-contractions
	and reversed $\beta$-contractions  and changes of bound variables.  That is,
	$P=_\beta Q$ iff \textbf{there exist} $P_0, \ldots, P_n$ ($n\geq 0$)  such that
	$P_0\equiv P$,  $P_n\equiv Q$,
	$(\forall i\leq n-1) (P_i\triangleright_{1\beta}P_{i+1}  \mbox{ or }P_{i+1}\triangleright_{1\beta}P_i  \mbox{ or } P_i\equiv_\alpha P_{i+1}).$
\end{definition}
\noindent (NB: equality with an \textbf{existential} force, which will show in the proof rules for the identity type.)

The same happens with $\lambda\beta\eta$-equality:\\
\begin{definition}($\lambda\beta\eta$-equality \cite{lambda})
	The equality-relation determined by the theory $\lambda\beta\eta$ is called
	$=_{\beta\eta}$; that is, we define
	$$M=_{\beta\eta}N\quad\Leftrightarrow\quad\lambda\beta\eta\vdash M=N.$$
\end{definition}

\begin{example}
	Take the term $M\equiv(\lambda x.(\lambda y.yx)(\lambda w.zw))v$. Then, it is $\beta\eta$-equal to $N\equiv zv$ because of the sequence:\\
	$(\lambda x.(\lambda y.yx)(\lambda w.zw))v, \quad  (\lambda x.(\lambda y.yx)z)v, \quad   (\lambda y.yv)z , \quad zv$\\
	which starts from $M$ and ends with $N$, and each member of the sequence is obtained via 1-step $\beta$- or $\eta$-contraction of a previous term in the sequence. To take this sequence into a {\em path\/}, one has to apply transitivity twice, as we do in the example below.
\end{example}

\begin{example}\label{examplepath}
	The term $M\equiv(\lambda x.(\lambda y.yx)(\lambda w.zw))v$ is $\beta\eta$-equal to $N\equiv zv$ because of the sequence:\\
	$(\lambda x.(\lambda y.yx)(\lambda w.zw))v, \quad  (\lambda x.(\lambda y.yx)z)v, \quad   (\lambda y.yv)z , \quad zv$\\
	Now, taking this sequence into a path leads us to the following:\\
	The first is equal to the second based on the grounds:\\
	$\eta((\lambda x.(\lambda y.yx)(\lambda w.zw))v,(\lambda x.(\lambda y.yx)z)v)$\\
	The second is equal to the third based on the grounds:\\
	$\beta((\lambda x.(\lambda y.yx)z)v,(\lambda y.yv)z)$\\
	Now, the first is equal to the third based on the grounds:\\
	$\tau(\eta((\lambda x.(\lambda y.yx)(\lambda w.zw))v,(\lambda x.(\lambda y.yx)z)v),\beta((\lambda x.(\lambda y.yx)z)v,(\lambda y.yv)z))$\\
	Now, the third is equal to the fourth one based on the grounds:\\
	$\beta((\lambda y.yv)z,zv)$\\
	Thus, the first one is equal to the fourth one based on the grounds:\\
	$\tau(\tau(\eta((\lambda x.(\lambda y.yx)(\lambda w.zw))v,(\lambda x.(\lambda y.yx)z)v),\beta((\lambda x.(\lambda y.yx)z)v,(\lambda y.yv)z)),\beta((\lambda y.yv)z,zv)))$.
\end{example}


The aforementioned theory establishes the equality between two $\lambda$-terms. Since we are working with computational objects as terms of a type, we need to translate the $\lambda\beta\eta$-equality to a suitable equality theory based on Martin L\"of's type theory. We obtain:

\begin{definition}
	The equality theory of Martin L\"of's type theory has the following basic proof rules for the $\Pi$-type:
	
	\bigskip
	
	\noindent
	\begin{bprooftree}
		\hskip -0.3pt
		\alwaysNoLine
		\AxiomC{$N : A$}
		\AxiomC{$[x : A]$}
		\UnaryInfC{$M : B$}
		\alwaysSingleLine
		\LeftLabel{$(\beta$) \quad}
		\BinaryInfC{$(\lambda x.M)N = M[N/x] : B[N/x]$}
	\end{bprooftree}
	\begin{bprooftree}
		\hskip 11pt
		\alwaysNoLine
		\AxiomC{$[x : A]$}
		\UnaryInfC{$M = M' : B$}
		\alwaysSingleLine
		\LeftLabel{$(\xi)$ \quad}
		\UnaryInfC{$\lambda x.M = \lambda x.M' : \Pi_{(x : A)}B$}
	\end{bprooftree}
	
	\bigskip
	
	\noindent
	\begin{bprooftree}
		\hskip -0.5pt
		\AxiomC{$M : A$}
		\LeftLabel{$(\rho)$ \quad}
		\UnaryInfC{$M = M : A$}
	\end{bprooftree}
	\begin{bprooftree}
		\hskip 100pt
		\AxiomC{$M = M' : A$}
		\AxiomC{$N : \Pi_{(x : A)}B$}
		\LeftLabel{$(\mu)$ \quad}
		\BinaryInfC{$NM = NM' : B[M/x]$}
	\end{bprooftree}
	
	\bigskip
	
	\noindent
	\begin{bprooftree}
		\hskip -0.5pt
		\AxiomC{$M = N : A$}
		\LeftLabel{$(\sigma) \quad$}
		\UnaryInfC{$N = M : A$}
	\end{bprooftree}
	\begin{bprooftree}
		\hskip 105pt
		\AxiomC{$N : A$}
		\AxiomC{$M = M' : \Pi_{(x : A)}B$}
		\LeftLabel{$(\nu)$ \quad}
		\BinaryInfC{$MN = M'N : B[N/x]$}
	\end{bprooftree}
	
	\bigskip
	
	\noindent
	\begin{bprooftree}
		\hskip -0.5pt
		\AxiomC{$M = N : A$}
		\AxiomC{$N = P : A$}
		\LeftLabel{$(\tau)$ \quad}
		\BinaryInfC{$M = P : A$}
	\end{bprooftree}
	
	\bigskip
	
	\noindent
	\begin{bprooftree}
		\hskip -0.5pt
		\AxiomC{$M: \Pi_{(x : A)}B$}
		\LeftLabel{$(\eta)$ \quad}
		\RightLabel {$(x \notin FV(M))$}
		\UnaryInfC{$(\lambda x.Mx) = M: \Pi_{(x : A)}B$}
	\end{bprooftree}
	
	\bigskip
	
\end{definition}

We are finally able to formally define computational paths:

\begin{definition}
	Let $a$ and $b$ be elements of a type $A$. Then, a \emph{computational path} $s$ from $a$ to $b$ is a composition of rewrites (each rewrite is an application of the inference rules of the equality theory of type theory or is a change of bound variables). We denote that by $a =_{s} b$.
\end{definition}

As we have seen in example \ref{examplepath}, composition of rewrites are applications of the rule $\tau$. Since change of bound variables is possible, each term is considered up to $\alpha$-equivalence.

\section{Identity Type}

In this section, we have two main objectives. The first one is to propose a formalization to the identity type using computational paths. The second objective is to show how can one use our approach to construct types representing reflexivity, transitivity and symmetry. In the case of the transitive type, we also compare our approach with the traditional one, i.e.,  Martin-L\"of's Intensional type. With this comparison, we hope to show  the clear advantage of our approach, in terms of simplicity. Since our approach is based on computational paths, we will sometimes refer to our formulation as the \emph{path-based} approach and the traditional formulation as the \emph{pathless} approach. By this we mean that, even though the Homotopy Type Theory approach to the identity type brings about the notion of paths in the semantics, there is little in the way of handling paths as terms in the language of type theory.

Before the deductions that build the path-based identity type, we would like to make clear that we will use the following construction of the traditional approach \cite{harper1}:

\bigskip
\begin{bprooftree}
	\AxiomC{$A$ type}
	\AxiomC{$a : A$}
	\AxiomC{$b : A$}
	\RightLabel{$Id- F$ \quad}
	\TrinaryInfC{$Id_{A}(a,b)$ type}
\end{bprooftree}
\begin{bprooftree}
	\AxiomC{$a : A$}
	\RightLabel{$Id - I$ \quad}
	\UnaryInfC{$r(a) : Id_{A}(a,a)$}
\end{bprooftree}
\bigskip
\begin{center}
	\begin{bprooftree}
		\alwaysNoLine
		\AxiomC{$a:A$}
		\AxiomC{$b:A$}
		\AxiomC{$c:Id_{A}(a,b)$}
		\AxiomC{$[x:A]$}
		\UnaryInfC{$q(x):C(x,x,r(x))$}
		\AxiomC{$[x:A,y:A,z:Id_{A}(x,y)]$}
		\UnaryInfC{$C(x,y,z)$ type}
		\RightLabel{$Id - E$ \quad}
		\alwaysSingleLine
		\QuinaryInfC{$J(p,q):C(a,b,c)$}
	\end{bprooftree}
\end{center}
\bigskip

\subsection{Path-based construction}

The best way to define any formal entity of type theory is by a set of natural deductions rules. Thus, we define our path-based approach as the following set of rules:

\begin{itemize}
	
	\item Formation and Introduction rules:
	
	\bigskip
	\begin{center}
		\begin{bprooftree}
			\AxiomC{$A$ type}
			\AxiomC{$a : A$}
			\AxiomC{$b : A$}
			\RightLabel{$Id - F$}
			\TrinaryInfC{$Id_{A}(a,b)$ type}
		\end{bprooftree}
		\begin{bprooftree}
			\AxiomC{$a =_{s} b : A$}
			\RightLabel{$Id - I$}
			\UnaryInfC{$s(a,b) : Id_{A}(a,b)$}
		\end{bprooftree}
	\end{center}
	\bigskip
	
	\item Elimination rule:
	
	\bigskip
	\begin{center}
		\begin{bprooftree}
			\alwaysNoLine
			\AxiomC{$m : Id_{A}(a,b)$ }
			\AxiomC{$[a =_{g} b : A]$}
			\UnaryInfC{$h(g) : C$}
			\alwaysSingleLine
			\RightLabel{$Id - E$}
			\BinaryInfC{$REWR(m, \acute{g}.h(g)) : C$}
		\end{bprooftree}
	\end{center}
	
	\bigskip
	
	\item Reduction rules:
	
	\bigskip
	\begin{center}
		\begin{bprooftree}
			\AxiomC{$a =_{m} b : A$}
			\RightLabel{$Id - I$}
			\UnaryInfC{$m(a,b) : Id_{A}(a,b)$}
			\alwaysNoLine
			\AxiomC{$[a =_{g} b : A]$}
			\UnaryInfC{$h(g) : C$}
			\alwaysSingleLine
			\RightLabel{$Id - E$ \quad $\rhd_\beta$}
			\BinaryInfC{$REWR(m, \acute{g}.h(g)) : C$}
		\end{bprooftree}
		\begin{bprooftree}
			\AxiomC{[$a =_{m} b : A$]}
			\alwaysNoLine
			\UnaryInfC{$h(m/g):C$}
		\end{bprooftree}
	\end{center}
	\bigskip

	\bigskip
	
	\begin{center}
		\begin{bprooftree}
			\AxiomC{$e : Id_{A}(a,b)$}
			\AxiomC{$[a =_{t} b : A]$}
			\RightLabel{$Id - I$}
			\UnaryInfC{$t(a, b) : Id_{A}(a, b)$}
			\RightLabel{$Id - E$ \quad  $\rhd_{\eta}$ \quad $e : Id_{A}(a,b)$}
			\BinaryInfC{$REWR(e, \acute{t}.t(a,b)) : Id_{A}(a,b)$}
		\end{bprooftree}
	\end{center}
	
	\bigskip
	
\end{itemize}

In these rules, $\acute{g}$ (and $\acute{t}$) to indicate that they are abstractions over the variable $g$ (or $t$), for which the main rules of conversion of $\lambda$-abstraction hold. For that reason, we proposed two reduction rules that handle these conversions, the $\beta$ and $\eta$ reduction rules.

Our introduction and elimination rules reassures the concept of equality as an \BF{existential force}. In the introduction rule, we encapsulate the idea that an witness of a identity type $Id_{A}(a,b)$ only exists if there exist a computational path establishing the equality of $a$ and $b$. Also, the elimination rule is similar to the elimination rule of the existential quantifier. If we have an witness for $Id_{A}(a,b)$, and if from a computational path between $a$ and $b$ we can construct a term of type $C$, then we can eliminate the identity type, obtaining a term of type $C$.

\section{A Term Rewriting System for Paths}

As we have just shown, a computational path establishes when two terms of the same type are equal. From the theory of computational paths, an interesting case arises. Suppose we have a path $s$ that establishes that $a =_{s} b : A$ and a path $t$ that establishes that $a =_{t} b : A$. Consider that $s$ and $t$ are formed by distinct compositions of rewrites. Is it possible to conclude that there are cases that $s$ and $t$ should be considered equivalent? The answer is \emph{yes}. Consider the following example:

\begin{example}
	\noindent \normalfont Consider the path  $a =_{t} b : A$. By the symmetric property, we obtain $b =_{\sigma(t)} a : A$. What if we apply the property again on the path $\sigma(t)$? We would obtain a path  $a =_{\sigma(\sigma(t))} b : A$. Since we applied symmetry twice in succession, we obtained a path that is equivalent to the initial path $t$. For that reason, we conclude the act of applying symmetry twice in succession is a redundancy. We say that the path $\sigma(\sigma(t))$ can be reduced to the path $t$.
\end{example}

As one could see in the aforementioned example, different paths should be considered equal if one is just a redundant form of the other. The example that we have just seen is just a straightforward and simple case. Since the equality theory has a total of 7 axioms, the possibility of combinations that could generate redundancies are high. Fortunately, all possible redundancies were thoroughly mapped by \cite{Anjo1}. In this work, a system that establishes all redundancies and creates rules that solve them was proposed. This system, known as $LND_{EQ}-TRS$, maps a total of 39 rules that solve redundancies. These 39 rules can be checked in \textbf{appendix B}. For each rule, there is a proof tree that constructs it. All proof trees can be checked in \cite{Ruy1}. In the case of example \textbf{3}, we have the following \cite{Ruy1}):

\bigskip
\begin{prooftree}
	\AxiomC{$x =_{t} y : A$}
	\UnaryInfC{$y =_{\sigma(t)} x : A$}
	\RightLabel{\quad $\rhd_{ss}$ \quad $x =_{t} y : A$}
	\UnaryInfC{$x =_{\sigma(\sigma(t))} y : A$}
\end{prooftree}

\bigskip

It is important to notice that we assign a label to every rule. In the previous case, we assigned the label $ss$.

\begin{definition}
	\normalfont An $rw$-rule is any of the rules defined in $LND_{EQ}-TRS$.
\end{definition}

\begin{definition}
	Let $s$ and $t$ be computational paths. We say that $s \rhd_{1rw} t$ (read as: $s$ $rw$-contracts to $t$) iff we can obtain $t$ from $s$ by an application of only one $rw$-rule. If $s$ can be reduced to $t$ by finite number of $rw$-contractions, then we say that $s \rhd_{rw} t$ (read as $s$ $rw$-reduces to $t$).
	
\end{definition}

\begin{definition}
	\normalfont  Let $s$ and $t$ be computational paths. We say that $s =_{rw} t$ (read as: $s$ is $rw$-equal to $t$) iff $t$ can be obtained from $s$ by a finite (perhaps empty) series of $rw$-contractions and reversed $rw$-contractions. In other words, $s =_{rw} t$ iff there exists a sequence $R_{0},....,R_{n}$, with $n \geq 0$, such that
	
	\centering $(\forall i \leq n - 1) (R_{i}\rhd_{1rw} R_{i+1}$ or $R_{i+1} \rhd_{1rw} R_{i})$
	
	\centering  $R_{0} \equiv s$, \quad $R_{n} \equiv t$
\end{definition}

\begin{proposition}\label{proposition3.7}  is transitive, symmetric and reflexive.
\end{proposition}

\begin{proof}
	Comes directly from the fact that $rw$-equality is the transitive, reflexive and symmetric closure of $rw$.
\end{proof}

We'd like to mention that  $LND_{EQ}-TRS$ is terminating and confluent. The proof of this affirmation can be found in \cite{Anjo1,Ruy2,Ruy3,RuyAnjolinaLivro}.

Thus, we conclude our review of computational paths as terms of the identity type and the associated rewrite system. If necessary, please check \cite{Ruy1} and \cite{Art3} for a thorough development of this theory.

\section{Homotopy Type Theory}	
In the previous sections, we have said that one of the most interesting concepts of type theory is the identity type. We have also said that the reason for that is the fact one can see the identity type as a homotopical path between two points of a space, giving rise to a homotopical interpretation of type theory. The connection between those two theories created a whole new area of research known as homotopy type theory. In this work, we introduced computational paths as the syntactic counterpart of those homotopical paths, since they only exist in a semantical sense. Nevertheless, we have not talked yet how one can use computational paths in homotopy type theory. Thus, in this section, we develop the main objective of this work. 

We want to show that some of the foundational definitions, propositions and theorems of homotopy type theory still hold in our path-based approach. In other words, we use our approach to construct the building blocks of more complex results. 

One important fact to notice is that every proof that does not involve the identity type is valid in the path-based approach. This is obvious, since the only difference between the traditional approach and ours is the formulation of the identity type. If a proof uses it, we need to reformulate this proof using our path-based approach, instead of using the induction principle of the traditional one. Thus, every part of a proof that is not directly or indirectly related to identity type is still valid in our approach.

In a path-based proof, we are going to use the formulation proposed in the previous sections. We also are going to use the reduction rules of $LND_{EQ}-TRS$.  In the process of developing the theory of this section, we noticed that $LND_{EQ}-TRS$, as proposed in the previous section is still incomplete. We state this based on the fact that we found new reduction rules that are not part of the original $LND_{EQ}-TRS$. That way, we added these new rules to the system, expanding it. 

\subsection{Groupoid Laws}

In our previous work \cite{Art3}, we have seen that computational paths form a groupoid structure. Let's check again those rules using our $REWR$ constructor directly:

\begin{lemma}
	The type  $\Pi_{(a : A)}Id_{A}(a,a)$ is inhabited. 
\end{lemma}

\begin{proof}
	We construct an witness for the desired type:
	
	\bigskip
	
	\begin{center}
		\begin{bprooftree}
			\AxiomC{$[a : A]$}
			\UnaryInfC{$a =_{\rho} a : A$}
			\RightLabel{$Id - I_{1}$}
			\UnaryInfC{$\rho(a,a) : Id_{A}(a,a)$}
			\RightLabel{$\Pi-I$}
			\UnaryInfC{$\lambda a.\rho(a,a) : \Pi_{(a : A)}Id_{A}(a,a)$}
		\end{bprooftree}
	\end{center}
	
	\bigskip
\end{proof}

\begin{lemma}
	The type  $\Pi_{(a : A)}\Pi_{(b : A)}(Id_{A}(a,b) \rightarrow Id_{A}(b,a))$ is inhabited.
\end{lemma}

\begin{proof}
	Similar to the previous lemma, we construct an witness:
	
	\bigskip
	
	\begin{center}
		\begin{bprooftree}
			\alwaysNoLine
			\AxiomC{$[a:A] \quad [b:A]$}
			\UnaryInfC{$[p(a,b) : Id_{A}(a,b)]$}
			\alwaysSingleLine
			\AxiomC{[$a =_{t} b : A$]}
			\UnaryInfC{$b =_{\sigma(t)} a : A$}
			\RightLabel{$Id - I$}
			\UnaryInfC{$(\sigma(t))(b,a) : Id_{A}(b,a)$}
			\RightLabel{$Id - E$}
			\BinaryInfC{$REWR(p(a,b),\acute{t}.(\sigma(t))(b,a)) : Id_{A}(b,a)$}
			\RightLabel{$\Pi-I$}
			\UnaryInfC{$\lambda p.REWR(p(a,b), \acute{t}.(\sigma(t))(b,a)) : Id_{A} (a,b) \rightarrow Id_{A}(b,a)$}
			\RightLabel{$\Pi-I$}
			\UnaryInfC{$\lambda b. \lambda p.REWR(p(a,b),\acute{t}.(\sigma(t))(b,a)) :  \Pi_{(b : A)}(Id_{A} (a,b) \rightarrow Id_{A}(b,a))$}
			\RightLabel{$\Pi-I$}
			\UnaryInfC{$\lambda a.\lambda b. \lambda p.REWR(p(a,b), \acute{t}.(\sigma(t))(b,a)) :  \Pi_{(a : A)}\Pi_{(b : A)}(Id_{A} (a,b) \rightarrow Id_{A}(b,a))$}
		\end{bprooftree}
		\bigskip
		
	\end{center}
\end{proof}

\begin{lemma}
	The type  $\Pi_{(a : A)}\Pi_{(b : A)}\Pi_{(c : A)} (Id_{A}(a,b) \rightarrow Id_{A}(b,c) \rightarrow Id_{A}(a,c))$ is inhabited.
\end{lemma}

\begin{proof}
	We construct the following witness:
	
	\begin{center}
		\begin{figure}
			\begin{sideways}
				\begin{bprooftree}
					\alwaysNoLine
					\AxiomC{$[a:A] \quad [b:A]$}
					\UnaryInfC{$[w(a,b) : Id_{A}(a,b)]$}
					\alwaysNoLine
					\AxiomC{$[c:A]$}
					\UnaryInfC{$[s(b,c) : Id_{A}(b,c)]$}
					\alwaysSingleLine
					\AxiomC{$[a =_{t} b:A]$}
					\AxiomC{$[b =_{u} c:A]$}
					\BinaryInfC{$a =_{\tau(t,u)} c:A$}
					\RightLabel{$Id - I$}
					\UnaryInfC{$(\tau (t,u))(a,c) : Id_{A}(a,c)$}
					\RightLabel{$Id - E$}
					\BinaryInfC{$REWR(s(b,c),\acute{u}(\tau (t,u))(a,c)) : Id_{A}(a,c)$}
					\RightLabel{$Id - E$}
					\BinaryInfC{$REWR(w(a,b),\acute{t}REWR(s(b,c),\acute{u}(\tau (t,u))(a,c))) : Id_{A}(a,c)$}
					\RightLabel{$\Pi-I$}
					\UnaryInfC{$\lambda s.REWR(w(a,b),\acute{t}REWR(s(b,c),\acute{u}(\tau (t,u))(a,c))) : Id_{A}(b,c) \rightarrow Id_{A}(a,c)$}
					\RightLabel{$\Pi-I$}
					\UnaryInfC{$\lambda w.\lambda s.REWR(w(a,b),\acute{t}REWR(s(b,c),\acute{u}(\tau (t,u))(a,c))) : Id_{A}(a,b) \rightarrow Id_{A}(b,c) \rightarrow Id_{A}(a,c)$}
					\RightLabel{$\Pi-I$}
					\UnaryInfC{$\lambda c.\lambda w.\lambda s.REWR(w(a,b),\acute{t}REWR(s(b,c),\acute{u}(\tau (t,u))(a,c))) :  \Pi_{(c : A)}(Id_{A}(a,b) \rightarrow Id_{A}(b,c) \rightarrow Id_{A}(a,c))$}
					\RightLabel{$\Pi-I$}
					\UnaryInfC{$\lambda b. \lambda c.\lambda w.\lambda s.REWR(w(a,b),\acute{t}REWR(s(b,c),\acute{u}(\tau (t,u))(a,c))) :  \Pi_{(b : A)}\Pi_{(c : A)}(Id_{A}(a,b) \rightarrow Id_{A}(b,c) \rightarrow Id_{A}(a,c))$}
					\RightLabel{$\Pi-I$}
					\UnaryInfC{$\lambda a. \lambda b. \lambda c.\lambda w.\lambda s.REWR(w(a,b),\acute{t}REWR(s(b,c),\acute{u}(\tau (t,u))(a,c))) :   \Pi_{(a : A)}\Pi_{(b : A)}\Pi_{(c : A)}(Id_{A}(a,b) \rightarrow Id_{A}(b,c) \rightarrow Id_{A}(a,c))$}
				\end{bprooftree}
			\end{sideways}
		\end{figure}
	\end{center}
	
	\bigskip
	
\end{proof}

\newpage

Lemmas \textbf{1}, \textbf{2} and \textbf{3} correspond respectively to the reflexivity, symmetry and transitivity of the identity type. From now on, the reflexivity will be represented by $\rho$, symmetry by $\sigma$ and transitivity by $\tau$.

\begin{lemma}
	For any type $A$, $x,y,z,w : A$ and $p : Id_{A}(x,y)$ and $q :  Id_{A}(y,z)$ and $r : Id_{A}(z,w)$, the following types are inhabited:
	
	\begin{enumerate}
		\item $\Pi_{(x,y : A)}\Pi_{(p : Id_{A}(x,y))} Id_{Id_{A}(x,y)}(p,\rho_{y} \circ p)$ and $\Pi_{(x,y : A)}\Pi_{(p : Id_{A}(x,y))}Id_{Id_{A}(x,y)}(p,p \circ \rho_{x})$.
		
		\item $\Pi_{(x,y : A)}\Pi_{(p : Id_{A}(x,y))}Id_{Id_{A}(x,y)}(\sigma(p) \circ p, \rho_{x})$ and $\Pi_{(x,y : A)}\Pi_{(p : Id_{A}(x,y))}Id_{Id_{A}(x,y)}(p \circ \sigma(p), \rho_{y})$
		
		\item $\Pi_{(x,y : A)}\Pi_{(p : Id_{A}(x,y))}Id_{Id_{A}(x,y)}(\sigma(\sigma(p)), p)$
		
		\item $\Pi_{(x,y,z,w : A)}\Pi_{(p : Id_{A}(x,y))}\Pi_{(q : Id_{A}(y,z))}\Pi_{(r : Id_{A}(z,w))}Id_{Id_{A}(x,w)}(r \circ (q \circ p), (r \circ q) \circ p)$ 
		
	\end{enumerate}
	
\end{lemma}

\begin{proof}
	The proof of each statement follows from the same idea. We just need to look for suitable reduction rules already present in the original $LND_{EQ}-TRS$.
	
	\begin{enumerate}
		\item The first thing to notice is that a composition in our path-based approach corresponds to a transitive operation, i.e.,  $(p \circ \rho_{x})$ can be written as $\tau(\rho_{x}, p)$  Follows from rules number \textbf{5} and \textbf{6}. These are as follows:
		
		\bigskip
		
		\begin{prooftree}
			\AxiomC{$x =_{r} y : A$}
			\AxiomC{$y =_{\rho} y : A$}
			\RightLabel{\quad $\rhd_{trr}$ \quad $x =_{r} y : A$}
			\BinaryInfC{$x =_{\tau(r,\rho)} y : A$}
		\end{prooftree}
		
		\bigskip
		
		\begin{prooftree}
			\AxiomC{$x =_{\rho} x : A$}
			\AxiomC{$x =_{r} y : A$}
			\RightLabel{\quad $\rhd_{tlr}$ \quad $x =_{r} y : A$}
			\BinaryInfC{$x =_{\tau(\rho,r)} y : A$}
		\end{prooftree}
		
		\bigskip
		
		Thus, we have:
		
		\bigskip
		
		\begin{prooftree}
			\AxiomC{$\tau(p,\rho_{y}) =_{trr} p : Id_{A}(x,y)$}
			\UnaryInfC{$(trr)(\tau(p,\rho_{y}), p) : Id_{Id_{A}(x,y)}(p, \rho_{y} \circ p)$}
			\UnaryInfC{$\lambda x.\lambda y.\lambda p.(trr)(\tau(p,\rho_{y}), p) : \Pi_{(x,y : A)}\Pi_{(p : Id_{A}(x,y))} Id_{Id_{A}(x,y)}(p,\rho_{y} \circ p)$}
			
		\end{prooftree}
		
		\bigskip
		
		\begin{prooftree}
			\AxiomC{$\tau(\rho_{x},p) =_{tlr} p : Id_{A}(x,y)$}
			\UnaryInfC{$(tlr)(\tau(\rho_{x},p), p) : Id_{Id_{A}(x,y)}(p,p \circ \rho_{x})$}
			\UnaryInfC{$\lambda x.\lambda y.\lambda p.(tlr)(\tau(\rho_{x}, p), p) : \Pi_{(x,y : A)}\Pi_{(p : Id_{A}(x,y))} Id_{Id_{A}(x,y)}(p,p \circ \rho_{x})$}
			
		\end{prooftree}	
		
		\bigskip
		
		\item  We use rules \textbf{3} and \textbf{4}:
		
		\begin{prooftree}
			\AxiomC{$x =_{r} y : A$}
			\AxiomC{$y =_{\sigma(r)} x : A$}
			\RightLabel{\quad  $\rhd_{tr}$ \quad $x =_{\rho} x : A$}
			\BinaryInfC{$x =_{\tau(r,\sigma(r))} x : A$}
		\end{prooftree}
		
		\bigskip
		
		\begin{prooftree}
			\AxiomC{$y =_{\sigma(r)} x : A$}
			\AxiomC{$x =_{r} y : A$}
			\RightLabel{\quad $\rhd_{tsr}$ \quad $y =_{\rho} y : A$}
			\BinaryInfC{$y =_{\tau(\sigma(r),r)} y : A$}
		\end{prooftree}
		
		\bigskip
		
		Thus:
		
		\bigskip
		
		\begin{prooftree}
			\AxiomC{$\tau(p,\sigma(p)) =_{tr} \rho_{x} : Id_{A}(x,y)$}
			\UnaryInfC{$(tr)(\tau(p,\sigma(p)),\rho_{x}) : Id_{Id_{A}(x,y)}(\sigma(p) \circ p, \rho_{x})$}
			\UnaryInfC{$\lambda x.\lambda y.\lambda p.(tr)(\tau(p,\sigma(p),\rho_{x}) : \Pi_{(x,y : A)}\Pi_{(p : Id_{A}(x,y))} Id_{Id_{A}(x,y)}(\sigma(p) \circ p, \rho_{x})$}
		\end{prooftree}
		
		\bigskip
		
		\begin{prooftree}
			\AxiomC{$\tau(\sigma(p),p) =_{tsr} \rho_{y} : Id_{A}(x,y)$}
			\UnaryInfC{$(tsr)(\tau(\sigma(p),p),\rho_{y}) : Id_{Id_{A}(x,y)}(p \circ \sigma(p),\rho_{y})$}
			\UnaryInfC{$\lambda x.\lambda y.\lambda p.(tsr)(\tau(p,\sigma(p),\rho_{y}) : \Pi_{(x,y : A)}\Pi_{(p : Id_{A}(x,y))} Id_{Id_{A}(x,y)}(p \circ \sigma(p), \rho_{y})$}
		\end{prooftree}
		
		\bigskip
		
		\item  We use rule \textbf{2}:
		
		\bigskip
		
		\begin{prooftree}
			\AxiomC{$x =_{r} y : A$}
			\UnaryInfC{$y =_{\sigma(r)} x : A$}
			\RightLabel{\quad $\rhd_{ss}$ \quad $x =_{r} y : A$}
			\UnaryInfC{$x =_{\sigma(\sigma(r))} y : A$}
		\end{prooftree}
		
		\bigskip
		
		Thus:
		
		\bigskip
		
		\begin{prooftree}
			\AxiomC{$\sigma(\sigma(p)) =_{ss} p :  Id_{A}(x,y)$}
			\UnaryInfC{$(ss)(\sigma(\sigma(p),p))  : Id_{Id_{A}(x,y)}(\sigma(\sigma(p)),p)$}
			\UnaryInfC{$\lambda x.\lambda y.\lambda p.(ss)(\sigma(\sigma(p),p)) :  \Pi_{(x,y : A)}\Pi_{(p : Id_{A}(x,y))} Id_{Id_{A}(x,y)}(\sigma(\sigma(p)),p)$}
		\end{prooftree}
		
		\bigskip
		
		\item  We use rule \textbf{37}:
		
		\bigskip
		
		\begin{prooftree}
			\hskip - 155pt
			\AxiomC{$x =_{t} y : A$}
			\AxiomC{$y =_{r} w : A$}
			\BinaryInfC{$x =_{\tau(t,r)} w : A$}
			\AxiomC{$w =_{s} z : A$}
			\BinaryInfC{$x =_{\tau(\tau(t,r),s)} z : A$}
		\end{prooftree}
		
		\begin{prooftree}
			\hskip 4cm
			\AxiomC{$x =_{t} y : A$}
			\AxiomC{$y=_{r} w : A$}
			\AxiomC{$w=_{s} z : A$}
			\BinaryInfC{$y =_{\tau(r,s)} z : A$}
			\LeftLabel{$\rhd_{tt}$}
			\BinaryInfC{$x =_{\tau(t,\tau(r,s))} z : A$}
		\end{prooftree}
		
		\bigskip
		
		Thus:
		
		\bigskip
		
		\begin{small}
			\begin{prooftree}	
				\AxiomC{$\tau(\tau(p,q),r) =_{tt} \tau(p,\tau(q,r)) :  Id_{A}(x,w)$}
				\UnaryInfC{$(tt)(\tau(\tau(p,q),r) =_{tt} \tau(p,\tau(q,r)))  : Id_{Id_{A}(x,w)}(r \circ (q \circ p), (r \circ q) \circ p)$}
				\UnaryInfC{$\lambda x.\lambda y.\lambda z.\lambda w.\lambda p.\lambda q. \lambda r. (ss)(\sigma(\sigma(p),p)) : \Pi_{(p : Id_{A}(x,y))}\Pi_{(q : Id_{A}(y,z))}\Pi_{(r : Id_{A}(z,w))}Id_{Id_{A}(x,w)}(r \circ (q \circ p), (r \circ q) \circ p)$}
			\end{prooftree}	
		\end{small}	
	\end{enumerate}
\end{proof}

\bigskip

With the previous lemma, we showed that our path-based approach yields the groupoid structure of a type up to propositional equality.

\subsection{Functoriality}

We want to show that functions preserve equality\cite{hott}.

\begin{lemma}
	The type $\Pi_{(x,y : A)}\Pi_{(f: A \rightarrow B)} (Id_{A}(x,y) \rightarrow Id_{B}(f(x),f(y)))$ is inhabited.
\end{lemma}

\begin{proof}
	It is a straightforward construction:
	
	\bigskip
	
	\begin{prooftree}
		\AxiomC{[$x =_{s} y : A$]}
		\AxiomC{[$f: A \rightarrow B$]}
		\BinaryInfC{$f(x) =_{\mu_{f}(s)} f(y) : B$}
		\UnaryInfC{$\mu_{f}(s)(f(x),f(y)) : Id_{B}(f(x),f(y))$}
		\AxiomC{$[p: Id_{A}(x,y)]$}
		\BinaryInfC{$REWR(p,\lambda s. \mu_{f}(s)(f(x),f(y))): Id_{B}(f(x),f(y))$}
		\UnaryInfC{$\lambda x. \lambda y. \lambda f. \lambda p.REWR(p,\lambda s. \mu_{f}(s)(f(x),f(y))):\Pi_{(x,y : A)}\Pi_{(f: A \rightarrow B)} (Id_{A}(x,y) \rightarrow Id_{B}(f(x),f(y)))$}
		
	\end{prooftree}
	
\end{proof}

\bigskip

\begin{lemma}
	For any functions $f: A \rightarrow B$ and $g: B \rightarrow C$ and paths $p: x =_{A} y$ and $q: y =_{A} z$, we have:
	
	\begin{enumerate}
		\item $\mu_{f}(\tau(p,q)) = \tau(\mu_{f}(p),\mu_{f}(q))$
		\item $\mu_{f}(\sigma(p)) = \sigma(\mu_{f}(p))$
		\item $\mu_{g}(\mu_{f}(p)) = \mu_{g \circ f}(p)$
		\item $\mu_{Id_{A}}(p) = p$
	\end{enumerate}
\end{lemma}

\begin{proof}
	\begin{enumerate}

		\item For the first time, we need to add a new rule to the original 39 rules of $LND_{EQ}-TRS$. We introduce rule \textbf{40}:
		
		\bigskip
		
		\begin{prooftree}
			\hskip -155pt
			\AxiomC{$x =_{p} y : A$}
			\AxiomC{$[f: A \rightarrow B]$}
			\BinaryInfC{$f(x) =_{\mu_{f}(p)} f(y) : B$}
			\AxiomC{$y =_{q} z : A$}
			\AxiomC{$[f: A \rightarrow B]$}
			\BinaryInfC{$f(y) =_{\mu_{f}(q)} f(z) : B$}	
			\BinaryInfC{$f(x) = \tau(\mu_{f}(p),\mu_{f}(q)) f(z) : B$}	
		\end{prooftree}
		
		\begin{prooftree}
			\hskip 5cm
			\AxiomC{$x =_{p} y : A$}
			\AxiomC{$y =_{q} z : A$}
			\LeftLabel{$\rhd_{tf}$}
			\BinaryInfC{$x =_{\tau(p,q)} z : A$}
			\AxiomC{$f : A \rightarrow B$}
			\BinaryInfC{$f(x) =_{\mu_{f}(\tau(p,q))} f(z) : B$}
		\end{prooftree}
		
		\bigskip
		
		Thus, we have $\mu_{f}(\tau(p,q)) =_{\sigma(tf)} \tau(\mu_{f}(p),\mu_{f}(q))$
		
		\item This one follows from rule \textbf{30}:
		
		\bigskip
		
		\begin{prooftree}
			\hskip -155pt
			\AxiomC{$x =_{p} y : A$}
			\AxiomC{$[f: A \rightarrow B]$}
			\BinaryInfC{$f(x) =_{\mu_{f}(p)} f(y) : B$}
			\UnaryInfC{$f(y) =_{\sigma(\mu_{f}(p))} f(x) : B$}
		\end{prooftree}
		
		\begin{prooftree}
			\hskip 4cm
			\AxiomC{$x =_{p} y : A$}
			\LeftLabel{$\rhd_{sm}$}
			\UnaryInfC{$y =_{\sigma(p)} x : A$}
			\AxiomC{$[f: A \rightarrow B]$}
			\BinaryInfC{$f(y) =_{\mu_{f}(\sigma(p))} f(x) : B$}
		\end{prooftree}
		
		\bigskip
		
		We have $\mu_{f}(\sigma(p)) =_{\sigma(sm)} \sigma(\mu_{f}(p))$
		
		\item We introduce rule \textbf{41}:
		
		\bigskip
		
		\begin{prooftree}
			\hskip -155pt
			\AxiomC{$x =_{p} y : A$}
			\AxiomC{$[f: A \rightarrow B]$}
			\BinaryInfC{$f(x) =_{\mu_{f}(p)} f(y) : B$}
			\AxiomC{$[g: B \rightarrow C]$}
			\BinaryInfC{$g(f(x)) =_{\mu_{g}(\mu_{f}(p))} g(f(y)) : C$}
		\end{prooftree}
		
		\begin{prooftree}
			\hskip 4cm
			\AxiomC{$x =_{p} y : A$}
			\AxiomC{[$x : A$]}
			\AxiomC{$[f: A \rightarrow B]$}
			\BinaryInfC{$f(x) : B$}
			\AxiomC{$[g: B \rightarrow C]$}
			\BinaryInfC{$g(f(x)) : C$}
			\UnaryInfC{$\lambda x. g(f(x)) \equiv (g \circ f) : A \rightarrow C$}
			\LeftLabel{$\rhd_{cf}$}
			\BinaryInfC{$g(f(x)) =_{\mu_{g \circ f}(p)} g(f(y)) : C$}
		\end{prooftree}
		
		\bigskip
		
		Then, $\mu_{g}(\mu_{f}(p)) =_{cf} \mu_{g \circ f}(p)$
		
		\item We introduce rule \textbf{42}:
		
		\bigskip
		
		\begin{prooftree}
			\AxiomC{$x =_{p} y : A$}
			\AxiomC{$[Id_{A} : A \rightarrow A]$}
			\BinaryInfC{$Id_{A}(x) =\mu_{Id_{A}(p)} Id_{A}(y) : A$}
			\RightLabel{\quad $\rhd_{ci}$ \quad $x =_{p} y : A$}	
			\UnaryInfC{$x =_{\mu_{Id_{A}}(p)} y : A$}
		\end{prooftree}
		
		\bigskip
		
		It follows that $\mu_{Id_{A}}(p) =_{ci} p$
	\end{enumerate}
\end{proof}

\subsection{Transport}

As stated in  \cite{Ruy5}, substitution can take place when no quantifier is involved. In this sense, there is a 'quantifier-less' notion of substitution. In type theory, this 'quantifier-less' substitution is given by a operation known as transport \cite{hott}. In our path-based approach, we formulate a new inference rule of 'quantifier-less' substitution \cite{Ruy5}:

\bigskip
\begin{prooftree}
	
	\AxiomC{$x =_{p} y : A$}
	\AxiomC{$f(x) : P(x)$}
	\BinaryInfC{$p(x,y)\circ f(x) : P(y)$}
	
\end{prooftree}

\bigskip

We use this transport operation to solve one essential issue of our path-based approach. We know that given a path $x =_{p} y : A$ and function $f: A \rightarrow B$, the application of axiom $\mu$ yields the path $f(x) =_{\mu_{f}(p)} f(y) : B$. The problem arises when we try to apply the same axiom for a dependent function $f : \Pi_{(x : A)} P(x)$. In that case, we want $f(x) = f(y)$, but we cannot guarantee that the type of $f(x) : P(x)$ is the same as $f(y) : P(y)$. The solution is to apply the transport operation and thus, we can guarantee that the types are the same:

\begin{center}
	
	\bigskip
	\begin{prooftree}
		\AxiomC{$x =_{p} y : A$}
		\AxiomC{$f : \Pi_{(x : A)} P(x)$}
		\BinaryInfC{$p(x,y) \circ f(x) =_{\mu_{f}(p)} f(y) : P(y)$}
	\end{prooftree}
\end{center}

\bigskip

\begin{lemma}
	(Leibniz's Law) The type $\Pi_{(x,y: A)} (Id_{A}(x,y) \rightarrow P(x) \rightarrow P(y))$ is inhabited.
\end{lemma}

\begin{proof}
	We construct the following tree:
	
	\bigskip
	
	\begin{prooftree}
		\AxiomC{$[x =_{p} y : A]$}
		\AxiomC{$[f(x) : P(x)]$}
		\BinaryInfC{$p(x,y)\circ f(x) : P(y)$}
		\UnaryInfC{$\lambda f(x). p(x,y) \circ f(x) : P(x) \rightarrow P(y)$}
		\AxiomC{$[z: Id_{A}(x,y)]$}
		\BinaryInfC{$REWR(z,\lambda p.\lambda f(x).p(x,y) \circ f(x)) : P(x) \rightarrow P(y)$}
		\UnaryInfC{$\lambda x. \lambda y.\lambda z.REWR(z,\lambda p.\lambda f(x).p(x,y) \circ f(x)) : \Pi_{(x,y: A)} (Id_{A}(x,y) \rightarrow P(x) \rightarrow P(y))$}
	\end{prooftree}
\end{proof}

\bigskip

The function $\lambda f(x). p(x,y) \circ f(x) : P(x) \rightarrow P(y)$ is usually written as $transport^{p}(p, -)$ and $transport^{p}(p,f(x)) : P(y)$ is usually written as $p_{*}(f(x))$. 

\begin{lemma}
	For any $P(x) \equiv B$, $x =_{p} y : A$ and $b : B$, there is a path $transport^{P}(p,b) = b$.
\end{lemma}

\begin{proof}
	The first to notice is the fact that in our formulation of transport, we always need a functional expression $f(x)$, and in this case we have only a constant term $b$. To address this problem, we consider a function $f = \lambda. b$ and then, we transport over $f(x) \equiv b$:
	
	\begin{center}
		$transport^{P}(p,f(x) \equiv b) =_{\mu(p)} (f(y) \equiv b)$. 
	\end{center}
	
	Thus, $transport^{P}(p,b) =_{\mu(p)} b$. We sometimes call this path $transportconst^{B}_{p}(b)$.
	
\end{proof}

\begin{lemma}
	For any $f: A \rightarrow B$ and $x =_{p} y : A$, we have 
	
	\begin{center}
		$\mu(p)(p_{*}(f(x)),f(y)) = \tau(transportconst^{B}_{p}, \mu_{f}(p))(p_{*}(f(x)),f(y))$
	\end{center}
\end{lemma}

\begin{proof}
	The first thing to notice is that in this case, $transportconst^{B}_{p}$ is the path $\mu(p)(p*(f(x),f(x))$ by lemma \textbf{8}. As we did to the rules of $LND_{EQ}-TRS$, we establishes this equality by getting to the same conclusion from the same premises by two different trees:

	In the first tree, we consider $f(x) \equiv b : B$ and transport over $b : B$:
	
	\bigskip
	\begin{prooftree}
		\AxiomC{$x =_{p} y : A$}	 	
		\AxiomC{$f(x) \equiv b : B$}
		\BinaryInfC{$p(x,y) \circ (f(x) \equiv b) : B$}
		\UnaryInfC{$p_{*}(f(x)) =_{\mu_{f}(p)} b \equiv f(x)$}
		\AxiomC{$x =_{p} y : A$}
		\AxiomC{$f: A \rightarrow B$}		
		\BinaryInfC{$f(x) =_{\mu_{f}(p)} f(y) : B$}
		\BinaryInfC{$p_{*}(f(x)) =_{\tau(\mu_{f}(p),\mu_{f}(p))} f(y) : B$}
	\end{prooftree}
	\bigskip
	
	In the second one, we consider $f(x)$ as an usual functional expression and thus, we transport the usual way:
	
	\bigskip
	\begin{prooftree}
		\AxiomC{$x =_{p} y : A$}
		\AxiomC{$f(x) : B$}
		\BinaryInfC{$p(x,y) \circ f(x) : B$}
		\UnaryInfC{$p_{*}(f(x)) =_{\mu_{f}(p)} f(y) : B$}
	\end{prooftree}
\end{proof}

\bigskip

\begin{lemma}
	For any $x =_{p} y : A$ and $q: y =_{A} z : A$, $f(x) : P(x)$, we have
	
	\begin{center}
		$q_{*}(p_{*}(f(x))) = (p\circ q)_{*}(f(x))$
	\end{center}
\end{lemma}

\begin{proof}
	We develop both sides of the equation and wind up with the same result:
	
	\begin{center}
		$q_{*}(p_{*}f(x)) =_{\mu(p)} q_{*}(f(y)) =_{\mu(q)} f(z)$
		
		$(p\circ q)_{*}(f(x)) =_{\mu(p \circ q)} f(z)$
	\end{center}
\end{proof}

\begin{lemma}
	For any $f: A \rightarrow B$, $x =_{p} y : A$ and $u: P(f(x))$, we have:
	
	\begin{center}
		$transport^{P \circ f}(p,u) = transport^{P}(\mu_{f}(p),u)$
	\end{center}
\end{lemma}

\begin{proof}
	This lemma hinges on the fact that there is two possible interpretations of $u$ that stems from the fact that $(g \circ f)(x) \equiv g(f(x))$. Thus, we can see $u$ as functional expression $g$ on $f(x)$ or an expression $g \circ f$ on $x$:
	
	\begin{center}
		\begin{figure}
			\begin{sideways}
				\begin{bprooftree}

					\AxiomC{$x =_{p} y :A$}
					\AxiomC{$ u \equiv (g \circ f)(x) : (P \circ f)(x)$}
					\BinaryInfC{$p(x,y) \circ (g \circ f)(x) : (P \circ f)(y)$}
					\UnaryInfC{$p(x,y) \circ (g \circ f)(x) =_{\mu(p)} (g \circ f)(y) : (P \circ f)(y)$}
					\UnaryInfC{$p(x,y) \circ (g \circ f)(x) =_{\mu(p)} g(f(y)) : P(f(y))$}
					
					\AxiomC{$x =_{p} y : A$}
					\UnaryInfC{$f(x) =_{u_{f}(p)} f(y) : B$}
					\AxiomC{$u \equiv g(f(x)) : P(f(x))$}
					\BinaryInfC{$\mu_{f}(p)(f(x),f(y)) \circ g(f(x)): P(f(y))$}
					\UnaryInfC{$\mu_{f}(p)(f(x),f(y)) \circ g(f(x)) =_{\mu(p)} g(f(y)): P(f(y))$}	
					\UnaryInfC{$ g(f(y)) =_{\sigma(\mu(p))} \mu_{f}(p)(f(x),f(y)) \circ g(f(x)) : P(f(y)) $}
					\BinaryInfC{$p(x,y) \circ (g \circ f)(x) =_{\tau(\mu(p),\sigma(\mu(p)))} \mu_{f}(p)(f(x),f(y)) \circ g(f(x)) : P(f(y))$}
					\UnaryInfC{$transport^{P \circ f}(p,u) =_{\tau(\mu(p),\sigma(\mu(p)))} transport^{P}(\mu_{f}(p),u)$}
					
				\end{bprooftree}
			\end{sideways}
		\end{figure}
	\end{center}
	
	\newpage
\end{proof}

\begin{lemma}
	For any $f: \Pi_{(x: A)}P(x) \rightarrow Q(x)$, $ x =_{p} y : A$ and $u(x) : P(x)$, we have:
	
	\begin{center}
		$transport^{Q}(p,f(u(x))) = f(transport^{P}(p,u(x)))$
	\end{center}
	
	\begin{proof}
		We proceed the usual way, constructing a derivation tree that establishes the equality:
		
		\bigskip
		
		\begin{prooftree}
			\AxiomC{$ x =_{p} y : A$}
			\AxiomC{$f(u(x)) : Q(x)$}
			\BinaryInfC{$p(x,y) \circ f(u(x)) : Q(y)$}
			\UnaryInfC{$p(x,y) \circ f(u(x)) =_{\mu(p)} f(u(y)) : Q(y)$}
			\AxiomC{$x =_{p} y : A$}
			\AxiomC{$u(x) : P(x)$}
			\BinaryInfC{$p(x,y) \circ u(x) : P(y)$}
			\UnaryInfC{$p(x,y) \circ u(x) =_{\mu(p)} u(y)  : P(y)$}
			\AxiomC{$f: \Pi_{(x: A)}P(x) \rightarrow Q(x)$}
			\BinaryInfC{$f(p(x,y) \circ u(x)) =_{\mu_{f}(\mu(p))} f(u(y)) : Q(y)$}
			\UnaryInfC{$f(u(y)) =_{\sigma(\mu_{f}(\mu(p)))} f(p(x,y) \circ u(x)) : Q(y) $}
			\BinaryInfC{$p(x,y) \circ f(u(x)) =_{\tau(\mu(p), \sigma(\mu_{f}(\mu(p))))}  f(p(x,y) \circ u(x)) $}
			\UnaryInfC{	$transport^{Q}(p,f(u(x))) =_{\tau(\mu(p), \sigma(\mu_{f}(\mu(p))))} f(transport^{P}(p,u(x)))$}
		\end{prooftree}
	\end{proof}
	\bigskip
	
\end{lemma}

\subsection{Homotopies}
In Homotopy Type Theory, a homotopy is defined as follows \cite{hott}:

\begin{definition}
	For any $f,g : \Pi_{(x: A)}P(x)$, a homotopy from $f$ to $g$ is a dependent function of type:
	
	\begin{center}
		$(f \sim g) \equiv \Pi_{(x : A)}(f(x) = g(x))$
	\end{center} 
\end{definition}

In our path-based approach, we have a homotopy $f,g : \Pi_{(x: A)}P(x)$ if for every $x: A$ we have a computational path between $f(x) = g(x)$. Thus, if we have a homotopy $H_{f,g} : f \sim g$, we derive the following rule:

\bigskip

\begin{center}
	\begin{prooftree}
		\AxiomC{$H_{f,g} : f \sim g$}
		\AxiomC{$f,g : \Pi_{(x: A)}P(x)$}
		\AxiomC{$x : A$}
		\TrinaryInfC{$f(x) =_{H_{f,g}(x)} g(x) : P(x)$}
	\end{prooftree}
\end{center}

\bigskip

And:

\bigskip

\begin{center}
	\begin{prooftree}
		\AxiomC{$f,g : \Pi_{(x: A)}P(x)$}
		\AxiomC{$x : A$}
		\alwaysNoLine
		\AxiomC{$[f,g: \Pi_{(x: A)}P(x), x : A]$}
		\UnaryInfC{$f(x) =_{p} g(x)$}
		\alwaysSingleLine
		\TrinaryInfC{$H^{p}_{f,g} : f \sim g$}
	\end{prooftree}
\end{center}

\bigskip

\begin{lemma}
	For any $f,g,h : A \rightarrow B$, the following types are inhabited:
	
	\begin{enumerate}
		\item $f \sim f$
		\item $(f \sim g) \rightarrow (g \sim f)$
		\item $(f \sim g) \rightarrow (g \sim h) \rightarrow (f \sim h)$
	\end{enumerate}
	
\end{lemma}

\begin{proof}
	\begin{enumerate}
		\item We construct the following term:
		
		\bigskip
		
		\begin{prooftree}
			\AxiomC{$f :A\rightarrow B$}
			\AxiomC{$ x : A$}
			\AxiomC{$[x : A]$}
			\UnaryInfC{$x =_{\rho} x$}
			\AxiomC{$[f: A \rightarrow B]$}
			\BinaryInfC{$f(x) =_{\mu_{f}(\rho)} f(x) : B$}
			\TrinaryInfC{$H^{\mu_{f}(\rho)}_{f,f} : f \sim f$}
		\end{prooftree}
		
		\bigskip
		
		\item We construct:
		
		\bigskip
		
		\begin{prooftree}
			\AxiomC{$f,g :A\rightarrow B$}
			\AxiomC{$ x : A$}	
			\AxiomC{$[H_{f,g} : f \sim g]$}
			\AxiomC{$[f,g : A \rightarrow B]$}
			\AxiomC{$[x : A]$}
			\TrinaryInfC{$f(x) =_{H_{f,g}(x)} g(x) : B$}
			\UnaryInfC{$g(x) =_{\sigma(H_{f,g}(x))} f(x) : B$}
			\TrinaryInfC{$H^{\sigma(H_{f,g}(x))}_{g,f} : g \sim f$}
			\UnaryInfC{$\lambda H_{f,g}. H^{\sigma(H_{f,g}(x))}_{g,f} : (f \sim g) \rightarrow (g \sim f)$}
		\end{prooftree}
		
		\bigskip
		
		\item We construct:
		
		\bigskip
		
		\begin{small}
			\begin{prooftree}
				\AxiomC{$f,h :A\rightarrow B$}
				\AxiomC{$x : A$}	
				\AxiomC{$[H_{f,g} : f \sim g]$}
				\AxiomC{$[f,g : A \rightarrow B]$}
				\AxiomC{$[x : A]$}	
				\TrinaryInfC{$f(x) =_{H_{f,g}(x)} g(x) : B$}	
				\AxiomC{$[H_{g,h} : g \sim h]$}
				\AxiomC{$[g,h : A \rightarrow B]$}
				\AxiomC{$[x : A]$}	
				\TrinaryInfC{$g(x) =_{H_{g,h}(x)} h(x) : B$}	
				\BinaryInfC{$f(x) =_{\tau(H_{f,g}(x),H_{g,z}(x))} h(x) :B$}
				\TrinaryInfC{$H^{\tau(H_{f,g}(x),H_{g,z}(x))}_{f,h} : f \sim h$}
				\UnaryInfC{$\lambda H_{f,g}. \lambda H_{g,h}.H^{\tau(H_{f,g}(x),H_{g,z}(x))}_{f,h} :(f \sim g) \rightarrow (g \sim h) \rightarrow (f \sim h) $}
			\end{prooftree}
		\end{small}		
	\end{enumerate}
\end{proof}

\bigskip

\begin{lemma}
	For any $H_{f,g} : f \sim g$ and functions $f,g : A \rightarrow B$ and a path $x =_{p} y : A$ we have:
	
	\begin{center}
		$\tau(H_{f,g}(x), \mu_{g}(p)) = \tau(\mu_{f}(p),H_{f,g}(y))$
	\end{center}
\end{lemma}

\begin{proof}
	
	To establish this equality, we need to add a new rule to our $LND_{EQ}-TRS$. We introduce rule \textbf{43}:
	
	\bigskip
	
	\begin{prooftree}
		\AxiomC{$H_{f,g} : f \sim g$}
		\AxiomC{$ x : A$}
		\AxiomC{$ f,g : A \rightarrow B$}
		\TrinaryInfC{$f(x) =_{H_{f,g}(x)} g(x) : B$}
		\AxiomC{$x =_{p} y : A$}
		\UnaryInfC{$g(x) =_{\mu_{g}(p)} g(y) : B$}
		\BinaryInfC{$f(x) =_{\tau(H_{f,g}(x), \mu_{g}(p))} g(y) : B$}
	\end{prooftree}
	
	\begin{center}
		$\rhd_{hp}$
	\end{center}
	
	\begin{prooftree}
		\AxiomC{$x =_{p} y : A$}
		\UnaryInfC{$f(x) =_{\mu_{f}(p)} f(y) : B$}
		\AxiomC{$H_{f,g} : f \sim g$}
		\AxiomC{$x : A$}
		\AxiomC{$f,g : A \rightarrow B$}
		\TrinaryInfC{$f(y) =_{H_{f,g}(y)} g(y) : B$}
		\BinaryInfC{$f(x) =_{\tau(\mu_{f}(p), H_{f,g}(y)} g(y) : B$}
	\end{prooftree}
	
	And thus:
	
	\begin{center}
		$\tau(H_{f,g}(x), \mu_{g}(p)) =_{hp} \tau(\mu_{f}(p),H_{f,g}(y))$
	\end{center}
\end{proof}

After this section, we start to study specific lemmas and theorems involving basic types of type theory. Nevertheless, several of those theorems are statements about the notion of \textbf{equivalence} (notation: $\simeq$). Before we define equivalence, we need the following definition \cite{hott}:

\begin{definition}
	A \textbf{quasi-inverse} of a function $f : A \rightarrow B$ is a triple $(g,\alpha, \beta)$ such that $g$ is a function $g: B \rightarrow A$ and $\alpha$ and $\beta$ are homotopies such that $\alpha : f \circ g \sim Id_{B}$ and $\beta: g \circ f \sim Id_{A}$
\end{definition}

A quasi-inverse of $f$ is usually written as $qinv(f)$.

\begin{definition}
	A function $f : A \rightarrow B$ is an equivalence if there is a quasi-inverse $qinv(f) : B \rightarrow A$.
\end{definition}

\subsection{Cartesian Product}

We start proving some important lemmas and theorems for the Cartesian product type. As we did in previous subsections, we proceed using our path-based approach. Before we prove our first theorem, it is important to remember that given a term $x : A \times B$, we can extract two projections, $FST(x) : A$ and $SND(x) : B$. Thus, given a path $x =_{p} y : A \times B$, we extract paths $FST(x) = SND(y) : A$ and $SND(x) = SND(y)  : B$.

\begin{theorem}
	The function $(x =_{p} y : A \times B) \rightarrow (FST(x) = FST(y) : A) \times (SND(x) = SND(y): B)$ is an equivalence for any $x$ and $y$.
\end{theorem}

\begin{proof}
	To show the equivalence, we need to show the following
	
	\begin{enumerate}
		\item From $x =_{p} y : A \times B$ we want to obtain $(FST(x) = FST(y) : A) \times (SND(x) = SND(y) : B)$ and from that, we want to go back to $x =_{p} y : A \times B$. 
		
		\item We want to do the inverse process. From  $(FST(x) = FST(y) : A) \times (SND(x) = SND(y) : B)$ we want to obtain $x =_{p} y : A \times B$ and then go back to $(FST(x) = FST(y) : A) \times (SND(x) = SND(y) : B)$.
	\end{enumerate}
	
	To show the first part, we need rule \textbf{21}:
	
	\bigskip
	
	\begin{prooftree}
		\AxiomC{$x =_{p} y : A \times B$}
		\UnaryInfC{$FST(x) =_{mu_{1}(p)} FST(y) : A$}
		\AxiomC{$ x =_{p} y : A \times B$}
		\UnaryInfC{$SND(x) =_{\mu_{2}(p)} SND(y) : B$}
		\BinaryInfC{$\langle FST(x), SND(x) \rangle =_{\epsilon(\mu_{1}(p),\mu_{2}(p))}  \langle FST(y), SND(y) \rangle : A \times B$}
	\end{prooftree}
	\begin{center}
		$\rhd_{mx}$ $x =_{p} y : A \times B$.
	\end{center}
	
	\bigskip
	
	Thus, applying rule $mx$ we showed the first part of our proof. For the second part, we need rules \textbf{14} and \textbf{15}:
	
	\bigskip
	
	\begin{prooftree}
		\AxiomC{$x =_{r} x' : A$}
		\AxiomC{$y =_{s} z : B$}
		\BinaryInfC{$\langle x,y \rangle =_{\epsilon_{\land}(r,s)} \langle x',z \rangle : A \times B$}
		\UnaryInfC{$FST(\langle x,y \rangle) =_{\mu_{1}(\epsilon_{\land}(r,s))} FST(\langle x',z \rangle) : A$}
	\end{prooftree}
	
	\begin{center}
		$\rhd_{mx2l}$ $x =_{r} x' : A$.
	\end{center}
	
	\bigskip
	
	And:
	
	\bigskip
	
	\begin{prooftree}
		\AxiomC{$x =_{r} y : A$}
		\AxiomC{$z =_{s} w : B$}
		\BinaryInfC{$\langle x,z \rangle =_{\epsilon_{\land}(r,s)} \langle y,w \rangle : A \times B$}
		\UnaryInfC{$FST(\langle x,z \rangle) =_{\mu_{2}(\epsilon_{\land}(r,s))} FST(\langle y,w \rangle) : B$}
	\end{prooftree}
	
	\begin{center}
		$\rhd_{mx2r}$ $z =_{s} w : B$.
	\end{center}
	
	\bigskip
	
	We also use the $\eta$-reduction for the Cartesian product:
	
	\begin{center}
		$\langle FST(x), SND(x) \rangle : A \times B \rhd_{\eta} x : A \times B$
	\end{center}
	
	We construct the following derivation tree:
	
	\bigskip
	
	\begin{prooftree}
		
		\AxiomC{$\langle FST(x) =_{s} FST(y), SND(x) =_{t} SND(y) \rangle$}
		\UnaryInfC{$FST(x) =_{s} FST(y) : A $}
		\AxiomC{$\langle FST(x) =_{s} FST(y), SND(x) =_{t} SND(y) \rangle$}
		\UnaryInfC{$SND(x) =_{t} SND(y) : B$}
		\BinaryInfC{$\langle FST(x), SND(x) \rangle =_{\epsilon_{\land}(s,t)} \langle FST(y), SND(y) \rangle : A \times B$}
		\RightLabel{$\rhd_{\eta}$}
		\UnaryInfC{$x =_{\epsilon(s,t)} y : A \times B$}
		
	\end{prooftree}
	
	\bigskip
	
	From $x =_{\epsilon(s,t)} y : A \times B$, we have:
	
	\bigskip
	
	\begin{prooftree}
		\AxiomC{$x =_{\epsilon(s,t)} y : A \times B$}
		\UnaryInfC{$FST(x) =_{\mu_{1}(\epsilon_{\land}(s,t)} FST(y) : A$}
		\AxiomC{$x =_{\epsilon(s,t)} y : A \times B$}
		\UnaryInfC{$SND(x) =_{\mu_{2}(\epsilon_{\land}(s,t)} SND(y) : B$}
		\RightLabel{$\land-I$}
		\BinaryInfC{$\langle FST(x) =_{\mu_{1}(\epsilon_{\land}(s,t)} FST(y), SND(x)=_{\mu_{2}(\epsilon_{\land}(s,t)} SND(y) \rangle$}
		\RightLabel{$\rhd_{mx2l,mx2r}$}
		\UnaryInfC{$\langle FST(x) =_{s} FST(y), SND(x) =_{t} SND(y) \rangle$}
		
	\end{prooftree}
	
	\bigskip
	
	Thus, we showed part 2 and concluded the proof of this theorem.
\end{proof}

\begin{theorem}
	For any type families $\Pi_{(z : Z)}A, \Pi_{(z: Z)}B$ and a type family defined by $(A \times B)(z) \equiv A(z) \times B(z)$, a path $  z =_{p} w : Z$ and $f(z) : A(z) \times B(z)$, we have:
	
	\begin{center}
		$transport^{A \times B}(p,f(z)) = \langle transport^{A}(p,FST(f(z))), transport^{B}(p,SND(f(z))) \rangle : A(w) \times B(w)$
	\end{center}
	
\end{theorem}

\begin{proof}
	We construct a derivation tree that establishes the equality:
	
	\bigskip
	
	\begin{center}
		\begin{figure}
			\begin{sideways}
				\begin{bprooftree}
					\AxiomC{$z =_{p} w : Z$}
					\AxiomC{$f(z) : A(z) \times B(z)$}
					\BinaryInfC{$p(z,w) \circ f(z) : A(w) \times B(w)$}
					\UnaryInfC{$p(z,w) \circ f(z) =_{\mu(p)} f(w) : A(w) \times B(w)$}
					\UnaryInfC{$p(z,w) \circ f(z) =_{\tau(\mu(p),\eta)} \langle FST(f(w)), SND(f(w)) \rangle : A(w) \times B(w)$}
					\AxiomC{$z =_{p} w : Z$}
					\AxiomC{$FST(f(z)) : A(z)$}
					\BinaryInfC{$p(z,w) \circ FST(f(z)) : A(w)$}
					\AxiomC{$z =_{p} w : Z$}
					\AxiomC{$SND(f(z)) : B(z)$}
					\BinaryInfC{$p(z,w) \circ SND(f(z)) : B(w)$}	
					\BinaryInfC{$\langle p(z,w) \circ FST(f(z)), p(z,w) \circ SND(f(z)) \rangle : A(w) \times B(w)$}
					\UnaryInfC{$\langle p(z,w) \circ FST(f(z)), p(z,w) \circ SND(f(z)) \rangle =_{\mu(p)} \langle FST(f(w)), SND(f(w)) \rangle$}
					\UnaryInfC{$\langle FST(f(w)), SND(f(w)) \rangle =_{\sigma(\mu(p))}\langle p(z,w) \circ FST(f(z)), p(z,w) \circ SND(f(z)) \rangle $}
					\BinaryInfC{$p(z,w) \circ f(z) =_{\tau(\tau(\mu(p),\eta),\sigma(\mu(p)))} \langle p(z,w) \circ FST(f(z)), p(z,w) \circ SND(f(z)) \rangle : A(w) \times B(w) $}
					\UnaryInfC{	$transport^{A \times B}(p,f(z)) =_{\tau(\tau(\mu(p),\eta),\sigma(\mu(p)))} \langle transport^{A}(p,FST(f(z))), transport^{B}(p,SND(f(z))) \rangle : A(w) \times B(w)$}
				\end{bprooftree}
			\end{sideways}
		\end{figure}
	\end{center}
	
	\newpage
\end{proof}

\begin{theorem}
	For any $x,y : A \times B$, $FST(x) =_{p} FST(y) : A$, $SND(x) =_{q} SND(y) : B$, functions $g : A \rightarrow A'$, $h: B \rightarrow B'$ and $f: A \times B \rightarrow A' \times B'$ defined by $f(x) \equiv \langle g(FST(x)),h(SND(x) \rangle$, we have:
	
	\begin{center}
		$\mu_{f}(\epsilon_{\land}(p,q)) = \epsilon_{\land} (\mu_{g}(p), \mu_{h}(q))$
	\end{center}
	
\end{theorem}

\begin{proof}
	We introduce rule \textbf{44}:
	
	\bigskip
	
	\begin{prooftree}
		\AxiomC{$FST(x) =_{p} FST(y) : A$}
		\AxiomC{$SND(x) =_{q} SND(y) : B$}
		\BinaryInfC{$\langle FST(x), SND(x) \rangle =_{\epsilon_{\land}(p,q)} \langle FST(y), SND(y) \rangle : A \times B$ }
		\RightLabel{$=_{\eta}$}
		\UnaryInfC{$x =_{\epsilon_{\land}(p,q)} y : A \times B$}
		\UnaryInfC{$f(x) =_{\mu_{f}(\epsilon_{\land}(p,q))} f(y) : A' \times B'$}
	\end{prooftree}
	
	\begin{center}
		$\rhd_{mxc}$
	\end{center}
	
	\bigskip
	
	\begin{prooftree}
		\AxiomC{$FST(x) =_{p} FST(y) : A$}
		\UnaryInfC{$g(FST(x)) =_{\mu_{g}(p)} g(FST(y)) : A'$}
		\AxiomC{$SND(x) =_{q} SND(y) : B$}
		\UnaryInfC{$h(SND(x)) =_{\mu_{h}(q)} h(SND(y)) : B'$}
		\BinaryInfC{$\langle g(FST(x), h(SND(x)) \rangle =_{\epsilon_{\land}(\mu_{g}(p),\mu_{h}(q))} \langle g(FST(y)), h(SND(y)) \rangle : A' \times B'$}
		\UnaryInfC{$f(x) =_{\epsilon_{\land}(\mu_{g}(p),\mu_{h}(q))} f(y) : A' \times B' $}
		
	\end{prooftree}
	\bigskip
	
	And thus:
	
	\begin{center}
		$\mu_{f}(\epsilon_{\land}(p,q)) =_{mxc} \epsilon_{\land} (\mu_{g}(p), \mu_{h}(q))$	
	\end{center}
\end{proof}

\subsection{Unit Type}

For the unit type $1$, our objective is to show the following theorem:

\begin{theorem}
	For any $x,y : 1$, there is a path $t$ such that $x =_{t} y$. Moreover, $t = \rho$.
\end{theorem}

\begin{proof}
	To show that there is such $t$, we need to use the induction for the unit type \cite{hott}:
	
	\begin{center}
		$* \rhd_{\eta} x : 1$
	\end{center}
	
	Therefore, given $x,y : 1$, we have:
	
	\bigskip
	
	\begin{prooftree}
		\AxiomC{$x =_{\sigma(\eta)} * : 1$}
		\AxiomC{$* =_{\eta} y : 1$}
		\BinaryInfC{$x =_{\tau(\sigma(\eta), \eta)} y : 1$}
	\end{prooftree}
	
	\bigskip
	
	Moreover, by rule \textbf{4}, we have:
	
	\begin{center}
		$\tau(\sigma(\eta),\eta) =_{tsr} \rho$. 
	\end{center}
	
	Thus, $t \equiv \tau(\sigma(\eta),\eta)$ and $t =_{tsr} \rho$.	
\end{proof}

\subsection{Function Extensionality}

In this subsection, we are interested in the property of function extensionality. In other words, we want to conclude that given any two functions $f,g$, if for any $x$ we have that $f(x) = g(x)$, then $f = g$. That $f = g$ implies $f(x) = g(x)$ by rules of basic type theory is shown in the sequel. Nonetheless, basic type theory is insufficient to derive function extensionality \cite{hott}. Our approach using computational paths also cannot derive full function extensionality. Nevertheless, we end up proving a weakened version which says that if $A$ is non-empty, then the above principle of function extensionality over $A\to B$ holds \cite{Ruy1}:
$$A\to (\Pi f^{A\to B}\Pi g^{A\to B}(\Pi x^A {\tt Id}_{B}(APP(f,x),APP(g,x))\to {\tt Id}_{A\to B}(f,g)))$$

The proof is as follows \cite{Ruy1}:

\begin{center}
	\begin{figure}
		\begin{sideways}
			\begin{bprooftree}
				\AxiomC{$[z : A]$}
				\AxiomC{$[f: A \rightarrow B]$}
				\alwaysNoLine
				\UnaryInfC{$[g : A \rightarrow B]$}
				\UnaryInfC{$[v:\Pi x^AId_B(APP(f,x),APP(g,x))]$}
				\alwaysSingleLine
				\BinaryInfC{$APP(v,z):Id_B(APP(f,z),APP(g,z))$}
				\AxiomC{$[f : A \rightarrow B]$}
				\UnaryInfC{$\lambda zAPP(f,z)=_\eta f:A\to B$}
				\UnaryInfC{$f=_{\sigma(\eta)}\lambda z.APP(f,z):A\to B$}
				\AxiomC{$[APP(f,z)=_t APP(g,z):B]$}
				\UnaryInfC{$\lambda zAPP(f,z)=_{\xi(t)} \lambda z.APP(g,z):A\to B$}
				\BinaryInfC{$f=_{\tau(\sigma(\eta),\xi(t))}\lambda z.APP(g,z):A\to B$}
				\AxiomC{$[g:A \rightarrow B]$}
				\UnaryInfC{$\lambda z.APP(g,z)=_\eta g:A\to B$}
				\BinaryInfC{$f=_{\tau(\tau(\sigma(\eta),\xi(t)),\eta)}g:A\to B$}
				\UnaryInfC{$(\tau(\tau(\sigma(\eta),\xi(t)),\eta))(f,g):Id_{A\to B}(f,g)$}
				\BinaryInfC{$REWR(APP(v,z),\acute{t}.(\tau(\tau(\sigma(\eta),\xi(t)),\eta))(f,g)):Id_{A\to B}(f,g)$}
				\UnaryInfC{$\lambda v.REWR(APP(v,z),\acute{t}.(\tau(\tau(\sigma(\eta),\xi(t)),\eta))(f,g)):\Pi x^A.Id_B(APP(f,x),AP(g,x))\to Id_{A\to B}(f,g)$}
				\UnaryInfC{$\lambda g.\lambda v.REWR(APP(v,z),\acute{t}.(\tau(\tau(\sigma(\eta),\xi(t)),\eta))(f,g)):\Pi g^{A\to B}(\Pi x^A.Id_B(APP(f,x),AP(g,x))\to Id_{A\to B}(f,g))$}
				\UnaryInfC{$\lambda f.\lambda g.\lambda v.REWR(AP(v,z),\acute{t}.(\tau(\tau(\sigma(\eta),\xi(t)),\eta))(f,g)):\Pi f^{A\to B}\Pi g^{A\to B}(\Pi x^A.Id_B(APP(f,x),AP(g,x))\to Id_{A\to B}(f,g))$}
				\UnaryInfC{$\lambda z.\lambda f.\lambda g.\lambda v.REWR(APP(v,z),\acute{t}.(\tau(\tau(\sigma(\eta),\xi(t)),\eta))(f,g)):A\to(\Pi f^{A\to B}\Pi g^{A\to B}(\Pi x^A.Id_B(APP(f,x),AP(g,x))\to Id_{A\to B}(f,g)))$}
			\end{bprooftree}
			
		\end{sideways}
	\end{figure}
\end{center}

\newpage 

Nevertheless, if we want full function extensionality and not just a weak version, we need to add a new rule to type theory. First, we let's prove the following lemma:

\begin{lemma}
	The following function exists:
	
	\begin{center}
		$(f = g) \rightarrow \Pi_{(x: A)}(f(x) = g(x) : B(x))$
	\end{center}
\end{lemma}

\begin{proof}
	The construction is straightforward:
	
	\bigskip
	\begin{prooftree}
		\AxiomC{$ [f =_{s} g]$}
		\AxiomC{$[x : A]$}
		\BinaryInfC{$f(x) =_{\nu(s)} g(x) : B(x)$}
		\UnaryInfC{$\lambda s. \lambda x.(f(x) =_{\nu(s)} g(x)) : (f = g) \rightarrow \Pi_{(x: A)}(f(x) = g(x) : B(x)) $}
	\end{prooftree}
\end{proof}

\bigskip

Now, to add function extensionality to our system, we need to add the following inference rule:

\begin{center}
	
	\bigskip
	
	\begin{prooftree}
		
		\AxiomC{$\lambda x.(f(x) =_{t} g(x)) : \Pi_{(x: A)}B$}
		\RightLabel{$ext$}
		\UnaryInfC{$f =_{ext(t)} g$}
	\end{prooftree}
	
	\bigskip
	
\end{center}

This rule is only needed if one wants to work with an extensional system. In that case, together with this inference rule, we also need to introduce two important reduction rules related to extensionality:

\begin{center}
	$ext(\nu(s)) =_{extr} s$
\end{center}

\begin{center}
	$\nu(ext(t)) =_{extl} t$
\end{center}

Since these rules are connected only to extensionality, we do not consider them as part of the basic rules of our rewriting system. Nevertheless, we can now prove the following:

\begin{lemma}
	$(f = g) \simeq \Pi_{(x: A)}(f(x) = g(x) : B(x))$
\end{lemma}

\begin{proof}
	
	This theorem is the direct application of the aforementioned extensionality rules. We have:
	
	\bigskip
	
	\begin{prooftree}
		\AxiomC{$f =_{s} g : \Pi_{(x : A)} B$}
		\AxiomC{$x : A$}
		\BinaryInfC{$f(x) =_{\nu(s)} g(x) : B(x)$}
		\RightLabel{$\rhd_{extr} \quad f =_{s} g : \Pi_{(x : A)}B$}
		\UnaryInfC{$\lambda x.(f(x) =_{\nu(s)} g(x)) : \Pi_{(x : A)}B$}
		\UnaryInfC{$f =_{ext(\nu(s))} g : \Pi_{(x : A)} B$}
	\end{prooftree}
	
	\bigskip
	
	We also have:
	
	\bigskip
	
	\begin{prooftree}
		\AxiomC{$\lambda x.(f(x) =_{t} g(x)) : \Pi_{(x: A)} B$}
		\UnaryInfC{$f =_{ext(t)} g : \Pi_{(x : A)} B$}
		\AxiomC{$[x : A]$}
		\RightLabel{$\rhd_{extl} \quad \lambda x.(f(x) =_{t} g(x)) : \Pi_{(x: A)} B$}
		\BinaryInfC{$f(x) =_{\nu(ext(t))} g(x) : B(x)$}
		\UnaryInfC{$\lambda x.(f(x) =_{\nu(ext(t))} g(x)) : \Pi_{(x: A)} B$}
	\end{prooftree}
	
	\bigskip

	Those two derivations tree establish the equivalence.
\end{proof}

Before we prove the next theorem, we need to revisit transport. For any function $f: A(x) \rightarrow B(x)$, it is possible to transport along this function $f$, resulting in $p_{*}(f) : A(y) \rightarrow B(y)$. In our approach, one should think of $p_{*}(f)$ as a function that has transport of a term $ a : A(x)$ as input, i.e., $p_{*}(a) : A(y)$. Thus, we define $p_{*}(f)$ point-wise:

\begin{center}
	$p_{*}(f)(p_{*}(a)) \equiv p_{*}(f(a))$
\end{center}

\begin{lemma}
	For any path $x =_{p} y : X$ and functions $f: A(x) \rightarrow B(x)$ and $g: A(y) \rightarrow B(y)$, we have the following equivalence:
	
	\begin{center}
		$(p_{*}(f) = g) \simeq \Pi_{(a : A(x))}(p_{*}(f(a)) = g(p_{*}(a)))$
	\end{center}
\end{lemma}

\begin{proof}
	We give two derivations tree, using the rules that we have established in the previous theorem:
	
	\bigskip
	
	\begin{prooftree}
		\AxiomC{$p_{*}(f) =_{p} g$}
		\AxiomC{$[a : A(x)]$}
		\BinaryInfC{$p_{*}(f)(p_{*}(a)) =_{\nu(p)} g(p_{*}(a)) : B(y)$}
		\UnaryInfC{$\lambda a.(p_{*}(f)(p_{*}(a) \equiv f(a)) =_{\nu(p)} g(p_{*}(a))) : \Pi_{(a : A(x))}(p_{*}(f(a)) = g(p_{*}(a))) $}
		\UnaryInfC{$p_{*}(f) =_{ext(\nu(p))} g$}
		\RightLabel{$\rhd_{extl}$}
		\UnaryInfC{$p_{*}(f) =_{p} g$}
	\end{prooftree}
	
	\bigskip
	
	And:
	
	\bigskip
	
	\begin{prooftree}
		\AxiomC{$\lambda a. (p_{*}(f(a)) =_{t} g(p_{*}(a)))$}
		\UnaryInfC{$\lambda a. (p_{*}(f)(p_{*}(a)) =_{t} g(p_{*}(a)))$}
		\UnaryInfC{$p_{*}(f) =_{ext(t)} g$}
		\AxiomC{$[a : A(x)]$}
		\BinaryInfC{$p_{*}(f)(p_{*}(a)) =_{\nu(ext(t))} g(p_{*}(a))$}
		\UnaryInfC{$p_{*}(f(a)) =_{\nu(ext(t))} g(p_{*}(a))$}
		\RightLabel{$\rhd_{extr}$}
		\UnaryInfC{$p_{*}(f(a)) =_{t} g(p_{*}(a))$}
		\UnaryInfC{$\lambda a. (p_{*}(f(a)) =_{t} g(p_{*}(a)))$}
	\end{prooftree}
	
	\bigskip
	
\end{proof}

\subsection{Univalence Axiom}

The first thing to notice is that in our approach the following lemma holds:

\begin{lemma}
	For any types $A$ and $B$, the following function exists:
	
	\begin{center}
		$idtoeqv: (A = B) \rightarrow (A \simeq B)$
	\end{center}
	
\end{lemma}

\begin{proof}
	The idea of the proof is similar to the one shown in \cite{hott}. We define $idtoeqv$ to be $p_{*} : A \rightarrow B$. Thus, to end this proof, we just need to show that $p_{*}$ is an equivalence.
	
	For any path $p$, we can form a path $\sigma(p)$ and thus, we have $(\sigma(p))_{*} : B \rightarrow A$. Now, we show that $(\sigma(p)))_{*}$ is a quasi-inverse of $p_{*}$. 
	
	We need to check that:
	
	\begin{enumerate}
		\item $p_{*}((\sigma(p)_{*}(b)) = b$
		\item $(\sigma(p))_{*}(p_{*}(a)) = a$
	\end{enumerate}
	
	Both equations can be shown by an application of lemma \textbf{5.10}:
	
	\begin{enumerate}
		\item $p_{*}((\sigma(p)_{*}(b)) = (\sigma(p) \circ p)_{*}(b) = \tau(p,\sigma(p))_{*}(b) =_{tr} \rho_{*}(b) =_{\mu(p)} b$.
		\item $(\sigma(p))_{*}(p_{*}(a)) = (p \circ \sigma(p))´_{*}(a) = \tau(\sigma(p),p)_{*}(a) =_{tsr} \rho_{*}(a) =_{\mu(p)} a$
	\end{enumerate}
\end{proof}

As we did in the previous section in lemma \textbf{5.15}, we showed that a function exists, but we did not show that it is an equivalence. In fact, basic type theory cannot conclude that $idtoeqv$ is an equivalence\cite{hott}. If we want this equivalence to be a property of our system, we must add a new axiom. This axiom is known as Voevodsky's univalence axiom\cite{hott}:

\begin{axiom}
	For any types $A,B$, $idtoeqv$ is an equivalence, i.e., we have:
	
	\begin{center}
		$(A = B) \simeq (A \simeq B)$
	\end{center}
	
\end{axiom}

\begin{lemma}
	For any $x,y : A$, $u(x) : B(x)$ and path $x =_{p} y : A$, we have:
	
	\begin{center}
		$transport^{B}(p,u(x)) = transport^{X \rightarrow X}(\mu_{B}(p),u(x)) = idtoeqv(\mu_{B}(p))(u(x))$
	\end{center}
\end{lemma}

\begin{proof}
	We develop every term of the equation and show that they arrive at the same conclusion:
	
	\bigskip
	
	\begin{prooftree}
		\AxiomC{$x =_{p} y : A$}
		\AxiomC{$u(x) : B(x)$}
		\BinaryInfC{$p(x,y) \circ u(x) : B(y)$}
		\UnaryInfC{$p(x,y) \circ u(x) =_{\mu(p)} u(y) : B(y)$}
	\end{prooftree}
	
	\bigskip
	
	\begin{prooftree}
		\AxiomC{$B(x) =_{\mu_{B}(p)} B(y)$}
		\AxiomC{$u(x) : B(x)$}
		\BinaryInfC{$\mu_{B}(p)(B(x),B(y)) \circ u(x) : B(y)$}
		\UnaryInfC{$\mu_{B}(p)(B(x),B(y)) \circ u(x) =_{\mu(p)} u(y): B(y)$}
	\end{prooftree}
	
	\bigskip
	
	Since $idtoeqv \equiv p_{*}$, we have that $idtoeqv(\mu_{B}(p))(u(x))$ is the same as $p_{*}(\mu_{B}(p))(u(x))$ that is the same as $transport^{X \rightarrow X}(\mu_{B}(p),u(x))$.	
\end{proof}

\subsection{Identity Type}

In this section, we investigate specific lemmas and theorems related to the identity type. We start with the following theorem:

\begin{theorem}
	if $f: A \rightarrow B$ is an equivalence, then for $x, y : A$ we have:
	
	\begin{center}
		$\mu_{f}: (x = y : A) \rightarrow (f(x) = f(y) : B)$
	\end{center}
\end{theorem}

\begin{proof}
	We will omit the specific details of this proof, since it is equal to the one of \textbf{theorem 2.11.1} presented in \cite{hott}. This is the case because this proof is independent of the usage of the induction principle of the identity type. The only difference is that at some steps we need to cancel inverse paths. In our approach, this is done by straightforward applications of \textbf{rules 3,4,5} and \textbf{6}.
\end{proof}

\begin{lemma}
	For any $a : A$, with $x_{1} =_{p} x_{2}$
	
	\begin{enumerate}
		
		\item $transport^{x \rightarrow (a = x)}(p,q(x_{1})) = \tau(q(x_{1}),p)$, \quad \quad \quad \quad \quad for $q(x_{1}): a = x_{1}$
		
		\item $transport^{x \rightarrow (x = a)}(p,q(x_{1})) = \tau(\sigma(p),q(x_{1})$, \quad \quad \quad \quad for $q(x_{1}): x_{1} = a$
		
		\item $transport^{x \rightarrow (x = x)}(p,q(x_{1})) = \tau(\sigma(p),\tau(q(x_{1}),p))$ \quad \ for $q(x_{1}) : x_{1} = x_{1}$
		
	\end{enumerate}
\end{lemma}

\bigskip

\begin{proof}
	\begin{enumerate}
		\item We start establishing the following reduction:
		
		\bigskip
		
		\begin{prooftree}
			\AxiomC{$a =_{q(x_{1})} x_{1}$}
			\AxiomC{$x_{1} =_{p} x_{2}$}
			\RightLabel{$\rhd \quad a =_{q(x_{2})} x_{2}$}
			\BinaryInfC{$ a =_{\tau(q(x_{1},p))} x_{2}$}
		\end{prooftree}
		
		\bigskip
		
		Thus, we just need to show that  $transport^{x \rightarrow (a = x)}(p,q(x_{1}))$ also reduces to $a =_{q(x(2))} x_{2}$:
		
		\bigskip
		
		\begin{prooftree}
			\AxiomC{$x_{1} =_{p} x_{2}$}
			\AxiomC{$q(x_{1}) : a = x_{1}$}
			\RightLabel{$=_{\mu(p)} \quad (a =_{q(x_{2})} x_{2})$}
			\BinaryInfC{$p(x_{1},x_{2}) \circ q(x_{1}) : a = x_{2}$}
		\end{prooftree}
		
		\bigskip
		
		\item We use the same idea:
		
		\bigskip
		
		\begin{prooftree}
			\AxiomC{$x_{2} =_{\sigma(p)} x_{1}$}
			\AxiomC{$x_{1} =_{q(x_{1})} a$}
			\RightLabel{$\rhd \quad x_{2} =_{q(x_{2})} a $}
			\BinaryInfC{$x_{2} =_{\tau(\sigma(p),q(x_{1}))} a$}
		\end{prooftree}
		
		\bigskip
		
		\begin{prooftree}
			\AxiomC{$x_{1} =_{p} x_{2}$}
			\AxiomC{$q(x_{1}) : x_{1} = a$}
			\RightLabel{$=_{\mu(p)} \quad (x_{2} =_{q(x_{2})} a)$}
			\BinaryInfC{$p(x_{1},x_{2}) \circ q(x_{1}) : x_{2} = a$}
		\end{prooftree}
		
		\bigskip
		
		\item Same as the previous cases:
		
		\bigskip
		
		\begin{prooftree}
			\AxiomC{$x_{2} =_{\sigma(p)} x_{1}$}
			\AxiomC{$x_{1} =_{q(x_{1})} x_{1}$}
			\BinaryInfC{$x_{2} =_{\tau(\sigma(p),q(x_{1}))} x_{1}$}
			\AxiomC{$x_{1} =_{p} x_{2}$}
			\RightLabel{$\rhd \quad x_{2} =_{q(x_{2})} x_{2}$}
			\BinaryInfC{$x_{2} =_{\tau(\tau(\sigma(p), q(x_{1})), p)} x_{2}$}
			
		\end{prooftree}
		
		\bigskip
		
		\begin{prooftree}
			\AxiomC{$x_{1} =_{p} x_{2}$}
			\AxiomC{$q(x_{1}) : x_{1} = x_{1}$}
			\RightLabel{$=_{\mu(p)} \quad (x_{2} =_{q(x_{2})} x_{2})$}
			\BinaryInfC{$p(x_{1},x_{2}) \circ q(x_{1}) : x_{2} = x_{2}$}
		\end{prooftree}	
	\end{enumerate}
\end{proof}

\bigskip

\begin{theorem}
	For any $f,g : A \rightarrow B$, with $a =_{p} a' : A$ and $f(a) =_{q(a)} g(a) : B$, we have:
	
	\begin{center}
		$transport^{x \rightarrow (f(x) = g(x)) : B}(p,q) = \tau(\tau(\sigma(\mu{f}(p)),q(a)), \mu_{g}(p)) : f(a') = g(a')$
	\end{center}
\end{theorem}

\begin{proof}
	This proof is analogous to the proof of the previous lemma:
	
	\bigskip
	
	\begin{prooftree}
		\AxiomC{$a =_{p} a': A$}
		\UnaryInfC{$f(a) =_{\mu_{f}(p)} f(a')$}
		\UnaryInfC{$f(a') =_{\sigma(\mu_{f}(p))} f(a)$}
		\AxiomC{$f(a) =_{q(a) g(a)}$}
		\BinaryInfC{$f(a) =_{\tau(\sigma(\mu_{f}(p)), q(a))} g(a)$}
		\AxiomC{$a =_{p} a'$}
		\UnaryInfC{$g(a) =_{\mu_{g}(p)} g(a')$}
		\RightLabel{$\rhd \quad f(a') =_{q(a')} g(a')$}
		\BinaryInfC{$f(a') =_{\tau(\tau(\sigma(\mu{f}(p)),q(a)), \mu_{g}(p))} g(a')$}
	\end{prooftree}
	
	\bigskip
	
	And:
	
	\bigskip
	
	\begin{prooftree}
		\AxiomC{$a =_{p} a'$}
		\AxiomC{$q(a) : f(a) = g(a)$}
		\RightLabel{$=_{\mu(p)} \quad (f(a') =_{q(a')} g(a'))$}
		\BinaryInfC{$p(a,a') \circ q(a) : f(a') = g(a')$}
	\end{prooftree}
\end{proof}

\bigskip

\begin{theorem}
	For any $f,g :\Pi_{(x : A)}B(x)$, with $a =_{p} a' : A$ and $f(a) =_{q(a)} g(a) : B(a)$, we have:
	
	\begin{center}
		$transport^{x \rightarrow (f(x) = g(x) : B(x))}(p,q) = \tau(\tau(\sigma(apd_{f}(p)),\mu_{transport^{B}p}(q)), apd_{g}(p))$
	\end{center}
	where $apd_{f}(p) \equiv (p(a,a') \circ f(a) =_{\mu(p)} f(a'))$ and $apd_{g} \equiv  (p(a,a') \circ g(a) =_{\mu(p)} g(a'))$
\end{theorem}

\begin{proof}
	Similar to previous theorem:
	
	\bigskip
	
	\begin{prooftree}
		\AxiomC{$p(a,a') \circ f(a) =_{\mu(p)} f(a')$}
		\UnaryInfC{$f(a') =_{\sigma(\mu(p))} p(a,a') \circ f(a)$}
		\AxiomC{$f(a) =_{q(a)} g(a)$}
		\UnaryInfC{$p(a,a') \circ f(a) =_{\mu_{trans^{B}p}(q(a))} p(a,a') \circ g(a)$}
		\BinaryInfC{$f(a') =_{\tau(\sigma(\mu(p)),\mu_{trans^{B}p}(q(a)))} p(a,a') \circ g(a)$}
		\AxiomC{$p(a,a') \circ g(a) =_{\mu(p)} g(a')$}
		\BinaryInfC{$f(a') =_{\tau(\tau(\sigma(\mu(p)),\mu_{trans^{B}p}(q(a))), \mu(p)} g(a')$}
	\end{prooftree}
	
	\bigskip
	
	\begin{center}
		$\rhd$ \quad $f(a') =_{q(a')} g(a')$
	\end{center}
	
	And:
	
	\bigskip
	
	\begin{prooftree}
		\AxiomC{$a =_{p} a'$}
		\AxiomC{$q(a) : f(a) = g(a)$}
		\RightLabel{$\rhd_{\mu(p)} \quad f(a') =_{q(a')} g(a')$}
		\BinaryInfC{$p(a,a') \circ q(a) : f(a') = g(a')$}
	\end{prooftree}	
\end{proof}

\bigskip

\begin{theorem}
	For any $a =_{p} a' : A$, $a =_{q} a$ and $a' =_{r} a'$, we have:
	
	\begin{center}
		$(transport^{x \rightarrow (x = x)}(p,q) = r) \simeq (\tau(q,p) = \tau(p,r))$
	\end{center}
	
\end{theorem}

\begin{proof}
	We use lemma \textbf{5.20} to prove this theorem, together with \textbf{rules 3,4,5,6} and \textbf{37}. We also consider functions $f(x) \equiv \tau(p,x) : (a' = z) \rightarrow (a = z)$ and $f^{-1}(x) \equiv \tau(\sigma(p),x) : (a = z) \rightarrow (a' = z)$. We proceed the same way as we have done to prove previous equivalences. In other words, we show two derivations trees. They are as follows:
	
	\bigskip
	
	\begin{prooftree}
		\AxiomC{$transport^{x \rightarrow (x = x)}(p,q) = r$}
		\RightLabel{lemma 5.20}
		\UnaryInfC{$\tau(\sigma(p),\tau(q,p)) = r$}
		\RightLabel{$\mu_{f}$}
		\UnaryInfC{$\tau(p,\tau(\sigma(p),\tau(q,p))) = \tau(p,r)$}
		\UnaryInfC{$\tau(\tau(p,\sigma(p)),\tau(q,p)) = \tau(p,r)$}
		\UnaryInfC{$\tau(\rho,\tau(q,p)) = \tau(p,r)$}
		\UnaryInfC{$\tau(q,p) = \tau(q,p)$}
	\end{prooftree}
	
	\bigskip
	
	And:
	
	\bigskip
	
	\begin{prooftree}
		\AxiomC{$\tau(q,p) = \tau(p,r)$}
		\RightLabel{$\mu_{f^{-1}}$}
		\UnaryInfC{$\tau(\sigma(p),\tau(q,p)) = \tau(\sigma(p),\tau(p,r))$}
		\UnaryInfC{$\tau(\sigma(p),\tau(q,p)) = \tau(\tau(\sigma(p),p),r)$}
		\UnaryInfC{$\tau(\sigma(p),\tau(q,p)) = \tau(\tau(\rho,r)$}
		\UnaryInfC{$\tau(\sigma(p),\tau(q,p)) = r$}
		\RightLabel{lemma 5.20}
		\UnaryInfC{$transport^{x \rightarrow (x = x)}(p,q) = r$}	
	\end{prooftree}
\end{proof}

\bigskip

\subsection{Coproduct}

One essential thing to remember is that a product $A + B$ has a left injection $inl : A \rightarrow A + B$ and $inr: B \rightarrow A + B$. As described in \cite{hott}, it is expected that $A + B$ contains copies of $A$ and $B$ disjointly. In our path based approach, we achieve this by constructing every path $inl(a) = inl(b)$  and $inr(a) = inr(b)$ by applications of axiom $\mu$ on paths $a = b$. Thus we show that we get the following equivalences:

\begin{enumerate}
	\item $(inl(a_{1}) = inl(a_{2})) \simeq (a_{1} = a_{2})$
	\item $(inr(b_{1}) = inr(b_{2})) \simeq (b_{1} = b_{2})$
	\item $(inl(a) = inr(b)) \simeq 0$
\end{enumerate}

To prove this, we use the same idea as in \cite{hott}. We characterize the type:

\begin{center}
	$(x \rightarrow (inl(a_{0}) = x)) : \Pi_{(x : A + B)}(inl(a_{0} = x))$
\end{center}

To do this, we define a type $code$:

\begin{center}
	$x : A + B \vdash code(x)$ type
\end{center}

Our main objective is to prove the equivalence $\Pi_{(x: A + B)}((inl(a_{0}) = x) \simeq code(x))$. Using the recursion principle of the coproduct, we can define $code$ by two equations:

\begin{center}
	$code(inl(a)) \equiv (a_{0} = a)$
	
	$code(inr(b)) \equiv 0$
\end{center}

\begin{theorem}
	For any $x : A + B$, we have $inl(a_{0} = x) \simeq code(x)$
\end{theorem}

\begin{proof}
	To show this equivalence, we use the same method as the one shown in \cite{hott}. The main idea is to define  functions
	
	\begin{center}
		
		$encode : \Pi_{(x : A + B)}\Pi_{(p : inl(a_{0}) = x)} code(x)$
		
		$decode: \Pi_{(x : A + B)} \Pi_{(c : code(x))}(inl(a_{0}) = x))$
		
	\end{center}
	
	such that $decode$ acts as a quasi-inverse of $encode$.
	
	We start defining $encode$:
	
	\begin{center}
		$encode(x,s) \equiv transport^{code}(s,\rho_{a_{0}})$
	\end{center}	
	
	We notice that $\rho_{a_{0}} : code(inl(a_{0}))$, since $code(inl(a_{0})) \equiv (a_{0} =_{\rho} a_{0})$
	We also notice that for $encode$, it is only possible for the argument $x$ to be of the form $x \equiv inl(a)$, since the other possibility is $x \equiv inr(a)$, but that case is not possible, because we would have a function to $code(inr(b)) \equiv 0$.
	
	For $decode$, when $x \equiv inl(a)$, we have that $code(x) \equiv a_{0} =_{c} a$ and thus, we define decode as $(inl(a_{0}) =_{\mu(c)} inl(a))$. When $x \equiv inr(a)$, then $code(x) \equiv 0$ and thus, we define $decode$ as having any value, given by the elimination of the type $0$. Now, we can finally prove the equivalence.
	
	Starting with $encode$, we have $x \equiv inl(a)$, $inl(a_{0}) =_{s} x$. Since \\	$encode(x,s) \equiv transport^{code}(s,\rho_{a_{0}})$, we have:
	
	\bigskip
	\begin{prooftree}
		\AxiomC{$inl(a_{0}) =_{s} inl(a)$}
		\AxiomC{$\rho_{a_{0}} : code(inl(a_{0}))$}
		\BinaryInfC{$s(inl(a_{0}),inl(a)) \circ \rho_{a_{0}} : code(inl(a))$}
		\RightLabel{$=_{\mu(s)}$}
		\UnaryInfC{$\rho_{a} : code(inl(a)) \equiv code(x)$}
	\end{prooftree}
	
	\bigskip
	
	Now, we can go back to $inl(a_{0}) = inl(a)$ by an application of $decode$, since:
	
	\begin{center}
		$decode(\rho_{a} : code(x)) \equiv inl(a_{0}) =_{\mu_{inl}} inl(a)$
	\end{center}
	
	And we conclude this part, since in our approach $inl(a_{0}) =_{s} inl(a)$ is constructed by applications of axiom $\mu$.
	
	Now, we start from decode. Let $c : code(x)$. If $x \equiv inl(a)$, then $c : a_{0} = a$ and thus, $decode(c) \equiv inl(a_{0}) =_{\mu(c)} inl(a)$. Now, we apply $encode$. We have:
	\bigskip
	
	\quad \quad \quad \quad \quad \quad \quad \quad \quad \quad $encode(x, \mu_{c}) = transport^{code}(\mu_{c},\rho_{a_{0}})$
	
	\quad \quad \quad \quad \quad \quad \quad \quad \quad \quad \quad \quad \quad \quad \quad \quad  $ = transport^{a \rightarrow (a_{0} = a)}(c,\rho_{a_{0}})$ \quad (\textbf{Lemma 11})
	
	\quad \quad \quad \quad \quad \quad \quad \quad \quad \quad \quad \quad \quad \quad \quad \quad   $
	= \tau(\rho_{a_{0}},c)$ \quad \quad \quad \quad \quad \quad \quad \  \quad(\textbf{Lemma 20})
	
	\quad \quad \quad \quad \quad \quad \quad \quad \quad \quad \quad \quad \quad \quad \quad \quad $= c$  \quad \quad \quad \quad \quad \quad \quad \quad \quad \quad \quad \ (\textbf{Rule 6})
	
	If $x \equiv inr(b)$, we have that $c : 0$ and thus, as stated in \cite{hott}, we can conclude anything we wish.
	
\end{proof}

\subsection{Reflexivity}

In this section, our objective is to conclude an important result related to the reflexive path $\rho$:

\begin{theorem}
	For any type A and a path $x =_{\rho} x : A$, if a path $s$ is obtained by a series (perhaps empty) of applications of axioms and rules of inference of $\lambda\beta\eta$-equality theory for type theory to the path $\rho$, then there is a path $t'$ such that $s =_{t'} \rho$.
\end{theorem}

\begin{proof}
	\begin{itemize}
		\item Base Case:
		
		We can start only with a path $x =_{\rho}$. In that case, it is easy, since we have $\rho =_{\rho} \rho$.

	\end{itemize}
	
	Now, we consider the inductive steps. Starting from a path $s$ and applying $\tau$, $\sigma$, we already have rules yield the desired path:
	
	\begin{itemize}
		\item $s = \sigma(s')$, with $s' =_{t'} \rho$.
		
		In this case, we have $s = \sigma(s') = \sigma(\rho) =_{sr} \rho$. 
		
		\item $s = \tau(s',s'')$, with $s' =_{t'} \rho$ and $s'' =_{t''} \rho$.
		
		We have that $s = \tau(s',s'') = \tau(\rho,\rho) =_{trr} \rho$ 
		
		The cases for applications of $\mu$, $\nu$ and $\xi$ remain to be proved. We introduce three new rules that handle these cases.
		
		\item $s = \mu(s')$, with $s' =_{t'} \rho$.
		
		We introduce rule \textbf{45}:
		
		\bigskip
		
		\begin{prooftree}
			\AxiomC{$x =_{\rho_{x}} x : A$}
			\AxiomC{$[f : A \rightarrow B]$}
			\RightLabel{$\rhd_{mxp}$ \quad $f(x) =_{\rho_{f(x)}} f(x) : B(x)$}
			\BinaryInfC{$f(x) =_{\mu(\rho_{x})} f(x) : B(x)$}
		\end{prooftree}
		
		\bigskip
		
		This rule is also valid for the dependent case:
		
		\bigskip
		
		\begin{prooftree}
			\AxiomC{$x =_{\rho_{x}} x : A$}
			\AxiomC{$[f : \Pi_{(x : A)} B(x)]$}
			\RightLabel{$\rhd_{mxp}$ \quad $f(x) =_{\rho_{f(x)}} f(x) : B(x)$}
			\BinaryInfC{$p(x,x) \circ f(x) =_{\mu(\rho_{x})} f(x) : B(x)$}
		\end{prooftree}	
		
		\bigskip
		
		Thus, we have $s = \mu(s') = \mu(\rho) =_{mxp} \rho$.
		
		\item $s = \nu(s')$, with $s' =_{t'} \rho$.
		
		We introduce rule \textbf{46}:
		
		\bigskip
		
		\begin{prooftree}
			\AxiomC{$f =_{\rho} f : \Pi_{(x:A)}B(x)$}
			\RightLabel{$\rhd_{nxp}$ \quad $f(x) =_{\rho_{f(x)}} f(x)$}
			\UnaryInfC{$f(x) =_{\nu(\rho_{x})} f(x) : B(x)$}		
		\end{prooftree} 
		
		\bigskip
		
		Thus, $s = \nu(s') = \nu(\rho) =_{nxp} \rho$.
		
		\item $s = \xi(s')$, with $s' =_{t'} \rho$.
		
		We introduce rule \textbf{47}:
		
		\bigskip
		
		\begin{prooftree}
			\AxiomC{$b(x) =_{\rho} b(x) : B$}
			\AxiomC{$x : A$}
			\RightLabel{$\rhd_{xxp}$ \quad $\lambda x.b(x) =_{\rho} \lambda x.b(x)$}
			\BinaryInfC{$\lambda x.b(x) =_{\xi(\rho)} \lambda x.b(x) : A \rightarrow B$}	
		\end{prooftree}
		
		\bigskip
		
		Thus, $s = \xi(s') = \xi(\rho) =_{xxp} \rho$.
		
		If we consider function extensionality, this theorem still holds:
		
		\item $s = ext(s')$, with $s' =_{t'} \rho$.
		
		We introduce a new rule to handle this case. Since it is related only to extensionality (i.e., when one admits the inference rule $ext$ to the system), we do not add this to the basic rules of our system.
		
		\bigskip
		
		\begin{prooftree}
			\AxiomC{$\lambda x.(f(x) =_{\rho} f(x)) : \Pi_{(x : A)}B(x)$}
			\RightLabel{$\rhd{exp}$ \quad $f =_{\rho} f$}
			\UnaryInfC{$ f =_{ext(\rho)} f$}
		\end{prooftree}
		
		\bigskip
		
		Thus, $s = ext(s') = ext(\rho) =_{exp} \rho$.
	\end{itemize}
\end{proof}

\subsection{Natural Numbers}

The Natural Numbers is a type defined inductively by an element $0 : \mathbb N$ and a function $succ : \mathbb N \rightarrow \mathbb N$. In our approach, the path space of the naturals is also characterized inductively. We start from the reflexive path $0 =_{\rho} 0$. All subsequent paths are constructed by applications of the inference rules of $\lambda\beta\eta$-equality. We show that this characterization is similar to the one constructed in \cite{hott}. To do this, we use $code$, $encode$ and $decode$. For $\mathbb N$, we define $code$ recursively \cite{hott}:

\begin{center}
	$code(0,0) \equiv 1$\\
	$code(succ(m),0) \equiv 0$\\
	$code(0,succ(m)) \equiv 0$\\
	$code(succ(m),succ(n)) \equiv code (m,n)$
\end{center}

We also define a dependent function $r : \Pi_{(n : \mathbb N)}code(m,n)$, with:

\begin{center}
	$r(0) \equiv *$\\
	$r(succ(n)) \equiv r(n)$
\end{center}

\begin{theorem}
	For any $m,n  : \mathbb N$, if there is a path $m =_{t} n : \mathbb N$, then $t \rhd \rho$.
\end{theorem}

\begin{proof}
	Since all paths are constructed from the reflexive path $0 =_{\rho} 0$, this is a direct application of theorem \textbf{5.10}.
\end{proof}

\begin{theorem}
	For any $m,n : \mathbb{N}$, we have $(m = n) \simeq code(m,n)$
\end{theorem}

\begin{proof}
	We need to define $encode$ and $decode$ and prove that they are quasi-inverses. We define $encode : \Pi_{(m,n : \mathbb N)}(m = n) \rightarrow code(m,n)$ as:
	
	\begin{center}
		$encode(m,n,p) \equiv transport^{code(m,-)} (p, r(m))$
	\end{center}
	
	We define $decode : \Pi_{(m,n : \mathbb N)}code(m,n) \rightarrow (m = n)$ recursively:
	
	\begin{center}
		$decode(0,0, c) \equiv 0 =_{\rho} 0$\\
		$decode(succ(m), 0, c) \equiv 0$\\
		$decode(0, succ(m), c) \equiv 0)$\\
		$decode(succ(m), succ(n), c) \equiv \mu_{succ}(decode(m,n,c))$
	\end{center}
	
	We now prove that if $m =_{p} n$, then $decode(code(m,n)) = \rho$. We prove by induction. The base is trivial, since $decode(0,0,c) \equiv \rho$. Now, consider $decode(succ(m),succ(n),c)$. We have that $decode(succ(m),succ(n),c) \equiv \mu_{succ}(decode(m,n,c))$. By the inductive hypothesis, $decode(m,n,c) \equiv \rho$. Thus, we need to prove that $\mu_{succ} = \rho$. This last step is a straightforward application of \textbf{rule 47}. Therefore,  $\mu_{succ} =_{mxp} \rho$. With this information, we can start the proof of the equivalence.
	
	For any $m =_{p} n$, we have:
	
	\begin{center}
		$encode(m,n,p) \equiv transport^{code(m,-)}(p,r(m))$
	\end{center}
	
	Thus:
	
	\bigskip
	
	\begin{prooftree}
		\AxiomC{$m =_{p} n$}
		\AxiomC{$r(m) : code(m,m)$}
		\RightLabel{$=_{\mu(p)} \quad (r(n) : code(m,n))$}
		\BinaryInfC{$p(m,n) \circ r(m) : code(m,n)$}
	\end{prooftree}
	
	\bigskip
	
	Now, we know that $decode(r(n) : code(m,n)) = \rho$ and,by theorem \textbf{5.11}, $p = \rho$.
	
	The proof starting from a $c : code(m,n)$ is equal to the one presented in \cite{hott}. We prove by induction. If $m$ and $n$ are $0$, we have the trivial path $0 =_{\rho} 0$, thus $decode(0,0,c) = \rho_{0}$, whereas $encode(0,0,\rho_{0}) \equiv r(0) \equiv *$. We conclude this part recalling that every $x : 1$ is equal to $*$, since we have $ x =_{\sigma(\eta)} * : 1$. In the case of $decode(succ(m),0,c)$ or $decode(0,succ(n),c)$, $c : 0$. The only case left is for $decode(succ(m), succ(n),c)$. Similar to \cite{hott}, we prove by induction:
	
	\bigskip
	
	$encode(succ(m),succ(n), decode(succ(m),succ(n),c))$
	
	\quad \quad \quad \quad \quad \quad \quad \quad \quad \quad $= encode(succ(m),succ(n), \mu_{succ}(decode(m,n,c))$
	
	\quad \quad \quad \quad \quad \quad \quad \quad \quad \quad $= transport^{code(succ(m),-)}(\mu_{succ}(decode(m,n,c)),r(succ(m))$
	
	\quad \quad \quad \quad \quad \quad \quad \quad \quad \quad $= transport^{code(succ(m),succ(-)}(decode(m,n,c),r(succ(m)))$
	
	\quad \quad \quad \quad \quad \quad \quad \quad \quad \quad $= transport^{code(m,-)}(decode(m,n,c),r(m))$
	
	\quad \quad \quad \quad \quad \quad \quad \quad \quad \quad $= encode(m,n,decode(m,n,c))$
	
	\quad \quad \quad \quad \quad \quad \quad \quad \quad \quad $= c$
\end{proof}

\subsection{Sets and Axiom K}

In this subsection, our objective is to prove, using our computational path approach, important results related to sets. First, We define the concept of set as done traditionally in Homotopy Type Theory. Then, we show that the connection between the axiom K and sets is also valid in our approach. We use these results to show that the naturals numbers are a set. We also prove Hedberg's theorem is valid in our theory, since only a few steps of its proof differs from method developed in \cite{hott}.

We start with the definition of set \cite{hott}:

\begin{definition}
	A type $A$ is a set if for all $x,y : A$ and all $p,q : x = y$, we have $p = q$.
\end{definition}

In type theory, if a type is a set, we also say that it has the uniqueness of identity proof (UIP) property, since all proofs of the equality of two terms $x = y$ are equal.

We now introduce the following axiom, known as \textbf{axiom K}\cite{Streicher2}:

\begin{center}
	For all $x : X$ and $p : (x =_{X} x)$, we have $p = refl_{x}$.
\end{center}

Of course, this previous formulation is one familiar to classic type theory. In our approach, axiom K can be understood as the following formulation:

\begin{axiom}
	For all $x : X$ and $x =_{t} x : X$, we have $t = \rho_{x}$.
\end{axiom}

Our objective is to establish a connection between sets and axiom K. The following lemma has proved to be useful:

\begin{lemma}
	For every path $t$, there is a path $t^{-1}$ such that $t \circ t^{-1} = \rho$ and $t^{-1} \circ t = \rho$. Furthermore, $t^{-1}$ is unique up to propositional identity.
\end{lemma}

\begin{proof}
	We claim that $t^{-1} = \sigma(t)$. The identities are straightforward. First, we have that $t \circ t^{-1} \equiv \tau(\sigma(t), t) =_{tsr} \rho$. We also have $t^{-1} \circ t \equiv \tau(t,\sigma(t)) =_{tr} \rho$. Now, suppose we have $s$ such that $\tau(t,s) =_{s'} \rho$. Thus, we have that $\tau(t,\sigma(t)) =_{\tau(tr,\sigma(s'))} \tau(t,s)$ and thus, $\sigma(t) = s$.
\end{proof}

\begin{theorem}
	A type $X$ is a set iff it satisfies axiom K.
\end{theorem}

\begin{proof}
	First, If $X$ is a set, we want to show that it satisfies axiom K. Suppose we have a path $x =_{t} x$. From the axioms of $\lambda\beta\eta$-equality, we also have that $x =_{\rho} x$. Since $X$ is a path, it is always the case that $t = \rho$ and thus, $X$ satisfies axiom K.
	
	If $X$ satisfies axiom K, we want to show that $X$ is a set. To do this, we want to show that given paths $p,q : x = y$, then we have $p = q$. We show this in the following manner. From a path $ x =_{q} y$, we apply $\sigma$ to obtain the inverse path $y =_{\sigma(q) x} x$. Then, we can concatenate $p$ and $\sigma(q)$, obtaining a path $x =_{\tau(p,\sigma(q))} x$. By axiom K, we have $\tau(p,\sigma(q)) = \rho$. Analogously, $\tau(\sigma(q),p) = \rho$. Thus, by the previous lemma, $\sigma(q) = p^{-1} = \sigma(p)$ and thus, $q = p$. Thus, by an application of $\sigma$, $p = q$. 
\end{proof}

\begin{theorem}
	$\mathbb N$ is a set.
\end{theorem}

\begin{proof}
	That $\mathbb N$ satisfies axiom K is a direct consequence of theorem \textbf{5.11}. Thus, from theorem \textbf{5.13}, we conclude that $\mathbb N$ is a set.
\end{proof}

In classic homotopy type theory, one can achieve the previous result by following a different path. Firstly, one should be aware of the following concept \cite{hott}:

\begin{definition}
	A type $X$ has \textbf{decidable equality} if for all $x, y : X$, the following type is inhabited:
	
	\begin{center}
		$(x = y : X) + \neg( x = y : X)$ 
	\end{center}
\end{definition}

\begin{theorem}
	If $X$ has decidable equality, then $X$ is a set.
\end{theorem}
This theorem is known as Hedberg's theorem . We will not show a full proof of it, since one can follow exactly the steps established in \cite{hott}. One should only be careful to notice that the path $apd$ of classic homotopy type theory is just our application of axiom $\mu$ on a dependent function $f$ and that \textbf{lemma 2.9.6} of \cite{hott} has already been proved in this work, in the form of lemma \textbf{5.17}.

We can also use Hedberg's theorem to give an alternative proof of theorem \textbf{5.14}, one similar to the one given in classic type theory:

\begin{theorem}
	$\mathbb N$ has decidable equality and thus, is a set.
\end{theorem}

\begin{proof}
	For any $x,y : \mathbb N$, we want to show that $(x = y) + \neg(x = y)$ is inhabited. We proceed by induction in $x$. For the base case, we have $x = 0$. If $y = 0$, then we have $0 =_{\rho} 0$. If $y = succ(n)$, we use the $encode$ type for $\mathbb N$. We have that $encode(0, succ(n)) :(0 = succ(n)) \rightarrow \textbf{0}$ and thus, $encode(0,succ(n)) : \neg(0 =succ(n))$.
	
	For the inductive step, we consider $ x = succ(m)$. If $y = 0$, then we can use $encode$ again to obtain $\neg(succ(m) = 0)$. If $y = succ(n)$, by the inductive hypothesis, we have two more cases to consider. If $m = n$, then we apply axiom $\mu$, and thus $succ(m) =_{\mu_{succ}} succ(n)$. If $\neg(m = n)$, then we just need to show that $succ$ is injective to obtain $\neg(succ(m) = succ(n))$. But that $succ$ is injective is a direct consequence of applying encode and then decode to $succ(m) = succ(n)$, since from that we conclude $m = n$. Therefore, from  $\neg(succ(m) = succ(n))$ we conclude $\neg(m = n)$.
\end{proof}

\subsection{Fundamental Group of a Circle}

The objective of this section is to show that it is possible to use computational paths to obtain one of the main results of homotopy theory, the fact that the fundamental group of a circle is isomorphic to the integers group. First, we define a circle as follows:

\begin{definition}[The circle $S^1$]
	A circle is the type generated by:
	
	\begin{itemize}
		\item A point $base : S^1$
		\item A computational path $base =_{loop} base : S^1$.
	\end{itemize}
\end{definition}

The first thing one should notice is that this definition doest not use only the points of the type $S^1$, but also a computational path $loop$ between those points. That is way it is called a higher inductive type \cite{hott}. Our approach differs from the classic one on the fact that we do not need to simulate the path-space between those points, since computational paths exist in the syntax of the theory. Thus, if one starts with a path  $base =_{loop} base : S^1$., one can naturally obtain additional paths applying the path-axioms $\rho$, $\tau$ and $\sigma$.  Thus, one has a path $\sigma(loop) = loop^{-1}$, $\tau(loop, loop)$, etc. In classic type theory, the existence of those additional paths comes from establishing that the paths should be freely generated by the constructors \cite{hott}. In our approach, we do not have to appeal for this kind of argument, since all paths comes naturally from direct applications of the axioms.

With that in mind, one can define the fundamental group of a circle. In homotopy theory, the fundamental group is the one formed by all equivalence classes up to homotopy of paths (loops) starting from a point $a$ and also ending at $a$. Since the we use computational paths as the syntax counterpart in type theory of homotopic paths, we use it to propose the following definition:

\begin{definition}[$\Pi_{1}(A,a)$ structure]
	$\Pi_{1}(A,a)$ is a structure defined as follows:
	
	\begin{center}
		$\Pi_{1}(A, a) = \{[loop]_{rw} \mid a =_{loop} a: A\}$
	\end{center}
\end{definition}

We use this structure to define the fundamental group of a circle. We also need to show that it is indeed a group.

\begin{proposition}
	$(\Pi_{1}(S,a), \circ)$ is a group.
\end{proposition}

\begin{proof}
	The first thing to define is the group operation $\circ$. Given any $a =_{r} a : S^1$ and $a =_{t} a : S^1$, we define $r \circ s$ as $\tau(s,r)$. Thus, we now need to check the group conditions:
	
	\begin{itemize}
		
		\item Closure: Given $a =_{r} a : S^1$ and $a =_{t} a : S^1$, $r \circ s$ must be a member of the group. Indeed, $r \circ s = \tau(s,r)$ is a computational path $a =_{\tau(s,r)} a : S^1$.
		\bigskip
		\item Inverse: Every member of the group must have an inverse. Indeed, if we have a path $r$, we can apply $\sigma(r)$. We claim that $\sigma(r)$ is the inverse of $r$, since we have:
		
		\begin{center}
			$\sigma(r) \circ r = \tau(r, \sigma(r)) =_{tr} \rho$
			
			$r \circ \sigma(r) = \tau(\sigma(r), r) =_{tsr} \rho$
		\end{center}
		
		Since we are working up to $rw$-equality, the equalities hold strictly.
		
		\item Identity: We use the path $a =_{\rho} a : S^1$ as the identity. Indeed, we have:
		
		\begin{center}
			$r \circ \rho = \tau(\rho,r) =_{tlr} r$
			
			$\rho \circ r = \tau(r,\rho) =_{trr} r$.
		\end{center}
		
		\item Associativity: Given any members of the group $a =_{r} a : S^1$, $a =_{t} a$ and $a =_{s} a$, we want that $r \circ (s \circ t) = (r \circ s) \circ t$:
		
		\begin{center}
			$r \circ (s \circ t) = \tau(\tau(t,s), r) =_{tt} \tau(t,\tau(s,r)) = (r \circ s) \circ t$
		\end{center}
		
	\end{itemize}
	
	All conditions have been satisfied. $(\Pi_{1}(S,a), \circ)$ is a group.
\end{proof}

Thus, 	$(\Pi_{1}(S,a), \circ)$ is indeed a group. We call this group the fundamental group of $S^1$. Therefore, the objective of this section is to show that $\Pi_{1}(S,a) \simeq \mathbb{Z}$.

Before we start developing this proof, the following lemma will prove to be useful:

\begin{lemma}
	All paths generated by a path $a =_{loop} a$ are $rw$-equal to a path $loop^{n}$, for a $n \in \mathbb Z$.
\end{lemma}

We have said that from a $loop$, one freely generate different paths applying the composition $\tau$ and the symmetry. Thus, one can, for example, obtain something such as $loop \circ loop \circ loop^{-1} \circ loop...$. Our objective with this lemma is to show that, in fact, this path can be reduced to a path of the form $loop^{n}$, for $n \in \mathbb Z$.

\begin{proof}
	The idea is to proceed by induction on the number $n$ of loops, i.e., $loop^{n}$. We start from a base $\rho$. For the base case, it is trivially true, since we define it to be equal to $loop^{0}$. From $\rho$, one can construct more complex paths by composing with $loop$ or $\sigma(loop)$ on each step.
	We have the following induction steps:
	
	\begin{itemize}
		\item A path of the form $\rho$ concatenated with $loop$: We have $\rho \circ loop = \tau(loop,\rho) =_{trr} loop = loop^{1}$;
		\bigskip
		\item A path of the form $\rho$ concatenated with $\sigma(loop)$: We have $\rho \circ \sigma(loop) = \tau(\sigma(loop),\rho) =_{trr} = \sigma(loop) = loop^{-1}$
		\bigskip
		\item A path of the form $loop^{n}$ concatenated with $loop$: We have $loop^{n} \circ loop = loop^{n+1}$.
		\bigskip
		\item A path of the form $loop^{n}$ concatenated with $\sigma(loop)$: We have $loop^{n} \circ \sigma(loop)$ $= (loop^{n-1} \circ loop) \circ \sigma(loop) =_{tt} loop^{n-1} \circ (loop \circ \sigma(loop)) =$ $loop^{n-1} \circ (\tau(\sigma(loop), loop)) =_{tsr} = loop^{n-1} \circ \rho = \tau(\rho, loop^{n-1}) =_{tlr} loop^{n-1}$
		\bigskip
		\item A path of the form $loop^{-n}$ concatenated with $loop$: We have $loop^{-n}$ = $loop^{-(n - 1)} \circ loop^{-1} = loop^{-(n - 1)} \circ \sigma(loop)$. Thus, we have $(loop^{-(n - 1)} \circ \sigma(loop)) \circ loop$ $=_{tt}$ $loop^{-(n - 1)} \circ (\sigma(loop) \circ loop)$ $=$ $loop^{-(n-1)} \circ \tau(loop,\sigma(loop)) =_{tr}$ $=$ $loop^{-(n-1)} \circ \rho = \tau(\rho,loop^{-(n-1)}) =_{tlr} loop^{-(n-1)}$.
		\bigskip
		\item a path of the form $loop^{-n}$ concatenated with $\sigma(loop)$: We have $loop^{-n} \circ loop^{-1} = loop^{-(n + 1)}$ 
	\end{itemize}
	
	Thus, every path is of the form $loop^{n}$, with $n \in \mathbb Z$.
\end{proof}

This lemma shows that every path of the fundamental group can be represented by a path of the form $loop^{n}$, with $n \in \mathbb Z$.

\begin{theorem}
	$\Pi_{1}(S,a) \simeq \mathbb{Z}$
\end{theorem}

To prove this theorem, one could use the approach proposed in \cite{hott}, defining an encode and decode functions. Nevertheless, since our computational paths are part of the syntax, one does not need to rely on this kind of approach to simulate a path-space, we can work directly with the concept of path.

\begin{proof}
	The proof is done by establishing a function from $\Pi_{1}(S,a)$ to $\mathbb{Z}$ and then an inverse from $\mathbb{Z}$ to $\Pi_{1}(S,a)$. Since we have access to the previous lemma, this task is not too difficult. The main idea is that the $n$ on $loop^{n}$ means the amount of times one goes around the circle, while the sign gives the direction (clockwise or anti-clockwise). In other words, it is the $winding$ number. Since we have shown that every path of the fundamental group is of the form $loop^{n}$, with $n \in \mathbb Z$, then we just need to translate $loop^{n}$ to an integer $n$ and an integer $n$ to a path $loop^{n}$. We define two functions, $toInteger: \Pi_{1}(S,a) \rightarrow \mathbb Z$ and $toPath: \mathbb Z \rightarrow \Pi_{1}(S,a)$:
	
	\begin{itemize}
		\item $toInteger$: To define this function, we use the help of two functions defined in $\mathbb Z$: the successor function $succ$ and the predecessor function $pred$. We define $toInteger$ as follows. Of course, we use directly the fact that every path of $\Pi_{1}(S,a)$ is of the form $loop^{n}$ with $n \in \mathbb Z$:
		
		\begin{equation*}
		toInteger: \begin{cases}
		toInteger(loop^n \equiv \rho) = 0 \quad \quad \quad \quad \quad \quad \quad \quad \quad \quad \enskip n = 0          \\
		toInteger(loop^{n}) = succ(toInteger(loop^{n-1})) \quad \quad n > 0  \\
		toInteger(loop^{n}) = pred(toInteger(loop^{n+1})) \quad \quad n < 0 \\
		\end{cases}
		\end{equation*}
		
		\item $toPath$: We just need to transform an integer $n$ into a path $loop^{n}$:
		
		\begin{equation*}
		toPath: \begin{cases}
		toPath(n) = \rho \quad \quad \quad \quad \quad \quad \quad \quad \quad\quad \enskip n = 0 \\
		toPath(n) = toPath(n - 1) \circ loop \quad \quad n > 0 \\
		toPath(n) = toPath(n + 1) \circ \sigma(loop) \quad n < 0 \\
		\end{cases}
		\end{equation*}
	\end{itemize}
	
	That they are inverses is a straightforward check. Therefore, we have $\Pi_{1}(S,a) \simeq \mathbb{Z}$.
\end{proof}

\subsection{Rules Added to $LND_{EQ}-TRS$}

In this section, we have introduced  $7$ new rules to the $LND_{EQ}-TRS$ system. It is the following list of rules:
\bigskip

\noindent 40. $\tau(\mu(r),\mu(s)) =_{tf} \mu(\tau(r,s))$\\
41. $\mu_{g}(\mu_{f}(p)) =_{cf} \mu_{g \circ f}(p)$\\
42. $\mu_{Id_{A}}(p) =_{ci} p$\\
43. $\tau(H_{f,g}(x), \mu_{g}(p)) =_{hp} \tau(\mu_{f}(p),H_{f,g}(y))$\\
44. $\mu_{f}(\epsilon_{\land}(p,q)) =_{mxc} \epsilon_{\land} (\mu_{g}(p), \mu_{h}(q))$\\	
45. $\mu_{f}(\rho_{x}) =_{mxp} \rho_{f(x)}$\\
46. $\nu(\rho_{x}) =_{nxp} \rho_{f(x)}$\\
47. $\xi(\rho) =_{xxp} \rho$\\

\bigskip

Moreover, if one adds extensionality to the theory, one winds up with three additional rules:
\bigskip

\noindent $\nu(ext(t)) =_{extl} t$\\
$\mu_{f}(\rho_{x}) =_{mxp} \rho_{f(x)}$\\
$ext(\rho) =_{exp} \rho$.

\section{Conclusion}
In this work, we connected our computational path approach to homotopy type theory. Using the algebra of computational paths, we have established important results of Homotopy Type Theory. That way, we have shown that our approach yields the main building blocks of Homotopy Type Theory, on par with the classic approach. We have also improved the rewrite system, adding new reduction rules. Indeed, we have ended this work with one of the most classic results of algebraic topology, the fact that the fundamental group of the circle is isomorphic to the group of the integers. 

In view of all results achieved in this work, we have developed a valid alternative approach to the identity type and homotopy type theory, based on this algebra of paths. We also believe that we have opened the way, in future works, for possible expansions of this results, formulating and proving even more intricate concepts and theorems of homotopy type theory using computational paths. 

\bibliographystyle{plain}
\bibliography{ref1}	

\begin{thebibliography}{10}

\bibitem{Steve1}
Steve Awodey.
\newblock Type theory and homotopy.
\newblock In P.~Dybjer, Sten Lindstr{\"o}m, Erik Palmgren, and G.~Sundholm,
  editors, {\em Epistemology versus Ontology}, volume~27 of {\em Logic,
  Epistemology, and the Unity of Science}, pages 183--201. Springer
  Netherlands, 2012.

\bibitem{Anjo1}
A.~G. de~Oliveira.
\newblock \emph{Proof transformations for labelled natural deduction via term
  rewriting}.
\newblock 1995.
\newblock Master's thesis, Depto. de Inform{\'a}tica, Universidade Federal de
  Pernambuco, Recife, Brazil, April 1995.

\bibitem{Ruy3}
A.~G. de~Oliveira and R.~J. G.~B. de~Queiroz.
\newblock A normalization procedure for the equational fragment of labelled
  natural deduction.
\newblock {\em Logic Journal of IGPL}, 7(2):173--215, 1999.

\bibitem{Ruy2}
R.~J. G.~B. de~Queiroz and A.~G. de~Oliveira.
\newblock Term rewriting systems with labelled deductive systems.
\newblock In {\em Proceedings of Brazilian Symposium on Artificial Intelligence
  (SBIA'94)}, pages 59--72, 1994.

\bibitem{Ruy5}
R.~J. G.~B. de~Queiroz and A.~G. de~Oliveira.
\newblock Natural deduction for equality: The missing entity.
\newblock In Luiz~Carlos Pereira, Edward Haeusler, and Valeria de~Paiva,
  editors, {\em Advances in Natural Deduction - A Celebration of Dag Prawitz's
  Work}, pages 63--91. Springer, 2014.

\bibitem{RuyAnjolinaLivro}
R.~J. G.~B. de~Queiroz, A.~G. de~Oliveira, and D.~M. Gabbay.
\newblock {\em The Functional Interpretation of Logical Deduction}.
\newblock World Scientific, 2011.

\bibitem{Ruy1}
R.~J. G.~B. de~Queiroz, A.~G. de~Oliveira, and A.~F. Ramos.
\newblock Propositional equality, identity types, and direct computational
  paths.
\newblock {\em South American Journal of Logic}, 2(2):245--296, 2016.
\newblock Special Issue A Festschrift for Francisco Miraglia, M. E. Coniglio,
  H. L. Mariano and V. C. Lopes (Guest Editors).

\bibitem{Ruy4}
R.~J. G.~B. de~Queiroz and D.~M. Gabbay.
\newblock Equality in labelled deductive systems and the functional
  interpretation of propositional equality.
\newblock In {\em Proceedings of the 9th Amsterdam Colloquium}, pages 547--565.
  ILLC/Department of Philosophy, University of Amsterdam, 1994.

\bibitem{harper1}
R.~Harper.
\newblock Type theory foundations, 2012.
\newblock Type Theory Foundations, Lecture at Oregon Programming Languages
  Summer School, Eugene, Oregon.

\bibitem{lambda}
J.~Roger Hindley and Jonathan~P. Seldin.
\newblock {\em Lambda-calculus and combinators: an introduction}.
\newblock Cambridge University Press, 2008.

\bibitem{Streicher2}
Martin Hofmann and Thomas Streicher.
\newblock The groupoid model refutes uniqueness of identity proofs.
\newblock In {\em Logic in Computer Science, 1994. LICS'94. Proceedings.,
  Symposium on}, pages 208--212. IEEE, 1994.

\bibitem{hofmann2}
Martin Hofmann and Thomas Streicher.
\newblock The groupoid interpretation of type theory.
\newblock In {\em Twenty-five years of constructive type theory ({V}enice,
  1995)}, volume~36 of {\em Oxford Logic Guides}, pages 83--111. Oxford Univ.
  Press, New York, 1998.

\bibitem{chenadec}
Philippe Le~Chenadec.
\newblock On the logic of unification.
\newblock {\em Journal of Symbolic computation}, 8(1):141--199, 1989.

\bibitem{Art3}
Arthur~F. Ramos, Ruy J. G.~B. De~Queiroz, and Anjolina~G. De~Oliveira.
\newblock On the identity type as the type of computational paths.
\newblock {\em Logic Journal of the IGPL}, 25(4):562--584, 2017.

\bibitem{hott}
The {Univalent Foundations Program}.
\newblock {\em Homotopy Type Theory: Univalent Foundations of Mathematics}.
\newblock \url{https://homotopytypetheory.org/book}, Institute for Advanced
  Study, 2013.

\bibitem{Vlad1}
V.~Voevodsky.
\newblock Univalent foundations and set theory, 2014.
\newblock Univalent Foundations and Set Theory, Lecture at IAS, Princeton, New
  Jersey, Mar 2014.

\end{thebibliography}

\newpage
\appendix

\section{Subterm Substitution}

In Equational Logic, the sub-term substitution is given by the following inference rule \cite{Ruy2}:
\begin{center}
	\begin{bprooftree}
		\AxiomC{$s = t$ }
		\UnaryInfC{$s\theta = t\theta$}
	\end{bprooftree}
\end{center}

One problem is that such rule does not respect the sub-formula property. To deal with that, \cite{chenadec} proposes two inference rules:

\begin{center}
	\begin{bprooftree}
		\AxiomC{$M = N$}
		\AxiomC{$C[N] = O$}
		\RightLabel{$IL$ \quad}
		\BinaryInfC{$C[M] = O$}
	\end{bprooftree}
	\begin{bprooftree}
		\AxiomC{$M = C[N]$}
		\AxiomC{$N = O$}
		\RightLabel{$IR$ \quad}
		\BinaryInfC{$M = C[O]$}
	\end{bprooftree}
\end{center}

\noindent where M, N and O are terms.

As proposed in \cite{Ruy1}, we can define similar rules using computational paths, as follows:

\begin{center}
	\begin{bprooftree}
		\AxiomC{$x =_r {\cal C}[y]: A$}
		\AxiomC{$y =_s u : A'$}
		\BinaryInfC{$x =_{{\tt sub}_{\tt L}(r,s)} {\cal C}[u]: A$}
	\end{bprooftree}
	\begin{bprooftree}
		\AxiomC{$x =_r w : A'$}
		\AxiomC{${\cal C}[w]=_s u : A$}
		\BinaryInfC{${\cal C}[x]=_{{\tt sub}_{\tt R}(r,s)} u : A$}
	\end{bprooftree}
\end{center}

\noindent where $C$ is the context in which the sub-term detached by '[ ]' appears and $A'$ could be a sub-domain of $A$, equal to $A$ or disjoint to $A$.

In the rule above, ${\cal C}[u]$ should be understood as the result of replacing every occurrence of $y$ by $u$ in $C$.

\section{List of Rewrite Rules}

We present the rewrite rules of $LND_{EQ}-TRS$. They are as follows (We show only the original 39 rules as proposed by \cite{Anjo1} and \cite{Ruy1}. The new rules added to the system appears in the end of section 5):
\\

\noindent 1. $\sigma(\rho) \triangleright_{sr} \rho$ \\ 
2. $\sigma(\sigma(r)) \triangleright_{ss} r$\\ 
3. $\tau({\cal C}[r] , {\cal C}[\sigma(r)]) \triangleright_{tr}  {\cal C }[\rho]$\\ 
4. $\tau({\cal C}[\sigma(r)], {\cal C}[r]) \triangleright_{tsr} {\cal C}[\rho]$\\ 
5. $\tau({\cal C}[r], {\cal C}[\rho]) \triangleright_{trr} {\cal C}[r]$\\ 
6. $\tau({\cal C}[\rho], {\cal C}[r]) \triangleright_{tlr} {\cal C}[r]$ \\ 
7. ${\tt sub_L}({\cal C}[r], {\cal C}[\rho]) \triangleright_{slr} {\cal C}[r]$\\ 
8. ${\tt sub_R}({\cal C}[\rho], {\cal C}[r]) \triangleright_{srr} {\cal C}[r]$ \\
9. ${\tt sub_L} ({\tt sub_L} (s, {\cal C}[r]), {\cal C}[\sigma(r)]) \triangleright_{sls} s$\\
10. ${\tt sub_L} ( {\tt sub_L} (s , {\cal C}[\sigma(r)]) , {\cal C}[r]) \triangleright_{slss} s$\\ 
11. ${\tt sub_R} ({\cal C}[s], {\tt sub_R} ({\cal C}[\sigma(s)],r)) \triangleright_{srs} r$\\ 
12. ${\tt sub_R} ({\cal C}[\sigma(s)], {\tt sub_R} ({\cal C}[s] ,  r )) \triangleright_{srrr} r$\\ 
13. 
$\mu_1 ( \xi_1 ( r))\triangleright_{mx2l1} r$\\
14. $\mu_1 ( \xi_\land ( r,s))\triangleright_{mx2l2} r$\\
15.
$\mu_2 ( \xi_\land ( r,s))\triangleright_{mx2r1} s$\\
16.
$\mu_2 ( \xi_2 ( s))\triangleright_{mx2r2} s$\\
17. 
$\mu ( \xi_1 (r) , s , u) \triangleright_{mx3l} s$\\ 
18. 
$\mu (\xi_2 (r) , s , u) \triangleright_{mx3r} u$\\ 
19.
$\nu (\xi (r)) \triangleright_{mxl} r$\\ 
20.
$\mu (\xi_2 (r) , s) \triangleright_{mxr} s$\\ 
21.
$\xi ( \mu_1 (r),\mu_2(r) ) \triangleright_{mx} r$ \\ 
22.
$\mu ( t, \xi_1 (r), \xi_2 (s)) \triangleright_{mxx} t$ \\ 
23. 
$\xi ( \nu (r) ) \triangleright_{xmr} r$ \\ 
24. 
$\mu (s,\xi_2 (r)) \triangleright_{mx1r} s$\\ 
25. $\sigma(\tau(r,s)) \triangleright_{stss} \tau(\sigma(s),  \sigma(r))$\\ 
26. $\sigma({\tt sub_L}(r,s)) \triangleright_{ssbl} {\tt sub_R}(\sigma(s), \sigma(r))$\\ 
27. $\sigma ({\tt sub_R} (r,s)) \triangleright_{ssbr} {\tt sub_L} (\sigma
(s),  \sigma (r))$\\ 
28. $\sigma(\xi (r)) \triangleright_{sx} \xi ( \sigma(r))$\\ 
29. $\sigma(\xi (s, r)) \triangleright_{sxss} \xi ( \sigma(s),  \sigma(r))$\\ 
30. $\sigma(\mu (r)) \triangleright_{sm} \mu ( \sigma(r))$\\ 
31. $\sigma(\mu (s, r)) \triangleright_{smss} \mu (\sigma(s),  \sigma(r))$\\ 
32. $\sigma(\mu (r,u,v)) \triangleright_{smsss} \mu ( \sigma(r),\sigma(u),\sigma(v))$\\
33. $\tau (r, {\tt sub_L} (\rho , s)) \triangleright_{tsbll} {\tt sub_L}  (r,s)$\\ 
34. $\tau (r, {\tt sub_R} (s, \rho)) \triangleright_{tsbrl}  {\tt 
	sub_L} (r,s)$\\ 
35. $\tau({\tt sub_L}(r,s),t) \triangleright_{tsblr} \tau (r, {\tt 
	sub_R} (s,t))$\\ 
36. $\tau ({\tt sub_R} (s,t),u) \triangleright_{tsbrr} {\tt sub_R} (s, \tau  (t,u))$\\ 
37. $\tau(\tau(t,r),s) \triangleright_{tt} \tau(t,\tau (r,s)) $\\
38. $\tau ({\cal C}[u], \tau ({\cal C}[\sigma(u)] , v)) \triangleright_{tts} v$\\
39. $\tau ({\cal C}[\sigma(u)] , \tau ({\cal C}[u] , v)) \triangleright_{tst} u$.

\bigskip

\end{document}